\documentclass[acmsmall, natbib=false]{acmart}
\pdfoutput=1 
\usepackage{amsmath}
\usepackage{appendix}
\usepackage{multirow}
\usepackage{graphicx}
\usepackage{bbm}
\usepackage{subcaption}

\usepackage{amsthm}
\usepackage{breqn}

\setcopyright{none}
\renewcommand\footnotetextcopyrightpermission[1]{} 
\begin{document}
\title{Ironing the Graphs: Toward a Correct Geometric Analysis of Large-Scale Graphs}

\author{Saloua Naama}

\email{saloua.naama@univ-smb.fr}
\affiliation{%
  \institution{Université Savoie Mont Blanc}
  \country{France}
}

\author{Kavé Salamatian}
\affiliation{%
  \institution{Université Savoie Mont Blanc}
  \country{France}}
\email{kave-salamatian@univ-smb.fr}

\author{Francesco Bronzino}
\affiliation{%
  \institution{École Normale Supérieure de Lyon}
  \country{France}
}
\email{francesco.bronzino@ens-lyon.fr}





\begin{abstract}
    Graph embedding approaches attempt to project graphs into geometric entities, {\em i.e.}, manifolds. At its core, the idea is that the geometric properties of the projected manifolds are helpful in the inference of graph properties. However, the choice of the embedding manifold is critical and, if incorrectly performed, can lead to misleading interpretations due to incorrect geometric inference. In this paper, we argue that the classical embedding techniques cannot lead to correct geometric interpretation as the microscopic details, {\em e.g.}, curvature at each point, of manifold, that are needed to derive geometric properties in  Riemannian geometry methods are not available, and we cannot evaluate the impact of the underlying space on geometric properties of shapes that lie on them. We advocate that for doing correct geometric interpretation the embedding of graph should be done over regular constant curvature manifolds. To this end, we present an embedding approach, the discrete Ricci flow graph embedding (dRfge) based on the discrete Ricci flow that adapts the distance between nodes in a graph so that the graph can be embedded onto a constant curvature manifold that is homogeneous and isotropic, {\em i.e.}, all directions are equivalent and distances comparable, resulting in correct geometric interpretations. A major contribution of this paper is that for the first time, we prove the convergence of discrete Ricci flow to a constant curvature and stable distance metrics over the edges. A drawback of using the discrete Ricci flow is the high computational complexity that prevented its usage in large-scale graph analysis. Another contribution of this paper is a new algorithmic solution that makes it feasible to calculate the Ricci flow for graphs of up to 50k nodes, and beyond. The intuitions behind the discrete Ricci flow make it possible to obtain new insights into the structure of large-scale graphs. We demonstrate this through a case study on analyzing the internet connectivity structure between countries at the BGP level. 
\end{abstract}

\fancyfoot{}

\maketitle
\thispagestyle{empty}

\newtheorem{lem}{Lemma}[section]

\section{Introduction}
Graphs are a fundamental tool to capture complex interactions through a relatively simple logical framework of
node-to-node relations. They are used in a wide range of applications and scientific domains \cite{Cai2018,berg2004}. Yet, understanding the structure and global properties of a graph remains challenging. Geometric interpretations are widely used to represent complex problems and help develop intuitions that lead to solutions. For example, such interpretations are at the core of classical machine learning techniques, like k-means, where a cluster is defined through a geometric centro\"id \cite{books/daglib/0033642}. Several attempts to define geometry for graphs are present in the literature. They are generally based on considering vertices as ``points'' seating over a low dimensional Riemannian manifold, and link weights as geodesic ``distances'' between these points ~\cite{ROBLESKELLY20071042, Cai2018}. This approach, called graph embedding, has a large set of applications including graph clustering \cite{LUO20032213}, links prediction and graph alignment \cite{berg2004}, graph robustness and resilience \cite{10.5555/979922.979965}, classification tasks in machine learning~\cite{4359350}. More recently, Graph Neural Networks (GNNs) solutions have relied on embeddings of nodes and links to project global graph properties into a space defined by the structure of a neural network \cite{Scarselli:2009ku}. 

It is noteworthy to highlight that assuming that link weights represent distances, as typically done in the literature \cite{10.5555/1564781, 10.5555/1202905}, does not immediately lead to correct geometric interpretations. Link weights might not even define a distance, thus invalidating triangular inequalities. Further, even when they represent a distance metric, by validating triangular inequalities, they are insufficient to validate geometric interpretations. Geometric properties result from the interactions of the shapes' properties ({\em e.g.}, distance, angles, areas, {\em etc.}), and the fundamental properties of the underlying space where the shapes are lying on ({\em e.g.} lengths of the sides of a triangle and its angles).

More advanced graph embedding techniques, like Multi-Dimensional Scaling (MDS)\cite{NIPS1994_1587965f} or isomap \cite{doi:10.1126/science.295.5552.7a}, project graphs into a manifold, but only provide vertex coordinates in the embedding space, and no microscopic (differential) properties of the embedding manifold. This means that from these embedding we cannot perform any geometric inference as we cannot evaluate the impact of the underlying space on shape properties. This leads to strong systematic biases and misleading geometric interpretations. Thus, considering the importance that geometric properties and interpretations have on various domains ({\em e.g.}, in machine learning techniques), we need to develop reliable methods to enable the correct geometric interpretations and inference. We advocate in this paper that geometric intuitions and ``correct''  geometric inference are only possible over regular constant curvature spaces, that are homogeneous and isotropic. 

The goal of this paper is to {\em develop a set of tools and methods for embedding a graph into a space where geometric inference approaches can be applied}. In particular, we leverage the Ricci Curvature to measure the local geometry of manifolds.  We use the continuous Ricci flow \cite{Hamilton88}, that acts on the surface of a manifold by deforming the distances around a point (stretching or squeezing them) proportionally to its Ricci curvature, {\em i.e.}, the Ricci flow acts like a steam iron pressing over the wrinkles of a bed linen (the initial manifold) and transforming it into a flat sheet (a constant curvature manifold). The Ricci flow is proven,  through the proof the Poincar{\'e} conjecture on 3-manifolds~\cite{Perelman2003}, to transform a given manifold into a constant curvature surface.  

Ollivier~\cite{OLLIVIER2007} extended the concept of Ricci curvature to graphs through the  Ollivier-Ricci curvature (ORC). This extension is based on a characterization of curvature of graphs in terms of optimal transport, {\em i.e.}, Wasserstein distance (see section \ref{sec:transport}).  In this paper, we extend this to graphs, through a non-linear diffusion recurrence, the Ricci flow originally defined for manifolds. We prove that the discrete Ricci flow recurrence converges to a fixed point where all links' curvature are constant and a corresponding ``canonic'' distance metric is assigned to each link. This theorem is the counterpart of the Perelman-Poincar{\'e} theorem \cite{Perelman2003} for graphs. Moreover, we prove the existence of an embedding manifold, with constant curvature given by the fixed point of the Ricci flow convergence, where graph vertices are positioned at ``canonic'' distances. This opens the way for a new type of graph embedding, the discrete Ricci flow graph embeddings (dRfge), that maps a graph into a regular and constant curvature manifold that is homogeneous and isotropic.  Over such a space geometric analysis becomes possible, resulting into a geometrization of the graph.

One shortcoming of our approach is that the underlying techniques used to perform the embedding necessitate the calculation of millions of shortest paths on the studied graphs. Thus, the complexity of the technique can quickly become intractable for larger graphs ({\em e.g}, 50k nodes or more). To tackle this challenge, we develop an algorithm approach that enables the calculation of the embedding in a reasonable time for large graphs. This algorithm reduces strongly the complexity by rearranging the edges during the calculation to leverage the existing redundancy when calculating a large number of the shortest paths. To demonstrate the usefulness of our approach, as well as to prove the practicality of the developed algorithms, we apply dRfge to analyze a commonly large graph used in networking, {\em i.e.}, the Board Gateway Protocol (BGP) graph.    



In summary, the contributions of this paper are the following:
\begin{enumerate}
  \item {\bf We prove that the discrete Ricci flow calculated over any graph converges}, resulting into projecting the graph into a constant curvature space.  (Section~\ref{sec:convergence&embedding}).
  \item {\bf We present a new algorithmic solution to calculate discrete curvatures on graphs.} Our algorithm enables the application of the discrete Ricci flow over large scale graphs in reasonable time, by reducing by two to three orders of magnitude the time needed to calculate the dRfge (Section~\ref{sec:algorithms}).
  \item Finally, {\bf we demonstrate how to use discrete Ricci flow distances to evaluate the robustness and the connectivity of graphs.} We show results for one case study: the structural analysis of Autonomous Systems (AS) level connectivity graphs extracted from BGP feeds (Section~\ref{sec:applications}). A second application to deriving robust spanning trees over road maps is pushed to the appendix.
\end{enumerate}

\section{Motivation and background}\label{sec:theory}
 
Geometry is one of the oldest areas of mathematics and has shaped our view of the world and our intuitions. A major application of geometry consists of calculating unknown distances and angles using known ones. Classical geometry deals with uniform spaces that validate two fundamental properties equivalent to the first four Euclid axioms \cite{culina_elementary_2018}:
(1)~homogeneity ({\em i.e.}, all points of space being equivalent) and (2)~ isotropy ({\em i.e.}, all directions being equivalent). There are only three spaces, Euclidean, hyperbolic, and spherical space, called regular constant curvature space, or Einstein space, that have these two properties. Over such spaces, one can easily evaluate the impact of the properties of the underlying space into the properties of shapes. For example, in the Euclidean space, we have the Al-Kashi theorem, or law of cosines, that relates the lengths of the sides of a triangle on an Euclidean plane to the cosine of one of its angles \cite{1891traite}. This formula, which is the basis of trigonometry, does not depend on the orientation of the triangle, nor on its position. Similar formulas exist in the hyperbolic and spherical spaces.


Several machine learning approaches on graphs depend on inferring unknown distances using known distances ({\em e.g.}, defining cluster centro{\"i}d in k-means, or projection into subspaces). Another tool that is central for geometric intuition is congruence, {\em i.e.}, matching one by one the points of two shapes through rotation and translation. However, congruence might not be valid over general manifolds that lack homogeneity and isotropy ({\em e.g.}, in a cylindrical space the distances along the cylinder and turning around the cylinder are not equivalent as geodesic curves cannot be superposed). In other terms, distances and angles are not enough to characterize shapes over general manifolds.

In Riemannian geometry, the properties of the underlying space are transferred to shapes through curvature.  Curvature indicates how much a geometric space deviates from a flat space, {\em i.e.}, an Euclidean space.  Differential geometry \cite{lafontaine} provides tools and methods coming from differential calculus to account for the impact of curvature on geometric properties like length, angles, areas, {\em etc.} For example, the distance between two very close points over the manifold, is calculated by projecting the manifold over a regular constant curvature space (a tangent plane) and readjusting the effect of curvature on the distance in the projected space. Geodesics, which are the shortest paths between two points over a manifold, are obtained through this infinitesimal readjustment process at all points, {\em i.e.}, geodesics are following lines in the tangent space that are zero curvature. In a nutshell, differential geometry extends concepts of Euclidean geometry to Riemannian spaces by finding the corrective terms, resulting from the shapes of underlying spaces, to be applied in each point. 

It is notable that laws similar to the law of cosine, which ties lengths and angles together in a triangle, only exist in homogeneous and isotropic spaces. The best that can be achieved for a non-constant curvature space is given by Toponogov's theorem \cite{lafontaine}, which states that in a negative curvature manifold (respectively, positive curvature space), the length of the third side of a triangle can only be upper bounded (respectively, lower bounded) by what predicted by the cosine law in a regular constant curvature space. In other terms, over a general manifold we cannot evaluate distances or angles using known ones, and classical geometric intuition through the cosine theorem, the classical $L_2$ distance,  will over-estimate or under-estimate distance depending on the curvature sign.  This strongly reduces the interest of embedding graphs into general manifolds as over them any intuition coming from classical geometry will be erroneous. 

In a nutshell, in the constant curvature space we can easily relate shapes' geometric properties irrespective of their orientation or position, while for the Riemannian space, we need to use, at each point, microscopic properties of the underlying space. This observation is the core motivation of this paper.


\section{Definitions}
\label{sec:definitions}
In this section, we present, in an incremental way, the concepts we will need for developing the discrete Ricci flow Graph Embedding (dRfge). Our aim here is to give the readers, who are not assumed to be knowledgeable about differential geometry, the intuition of the mathematical concepts behind the practical tools we will introduce in the paper. We have therefore balanced the mathematical precision with the ease of reading. A complete and precise presentation of differential and Riemannian geometry materials can be found in \cite{lafontaine} and in the given references there. Moreover, most of these concepts presented have already been defined in detail in the given references. Our main contribution is to prove the existence of a constant curvature embedding for graphs, to provide an iterative method that converges to the embedding, and to prove its convergence. In this section, we will introduce the set of concepts and parameters that will result into the final embedding with constant curvature. In the forthcoming, we will consider a weighted graph $G=(V, E, w)$. We are not assuming that the weights $w_{st}$ are a distance and they might not validate triangular inequality. In the simplest case, the weights are 1 for all edges. 
\subsection{Optimal Transport in graphs} \label{sec:transport}
Wasserstein distance and Optimal Transport over a graph have been extensively used in the literature \cite{Villani, pmlr-v97-titouan19a, Kolouri2020WassersteinEF,10.1007/978-3-030-67661-2_34}. It can be describe through supposing a set of $m$ sources  $\mathcal{S}=\{x_i\}_{i=1}^{m}\in V$ that produces each a share $\mu(x_i)$ of a commodity, and a set of $n$ sinks $\mathcal{D}=\{y_j\}_{j=1}^{n}\in V$ that want to consume each a share $\nu(y_j)$ of the commodity. The minimal cost of transferring one unit of the commodity from $x_i$ to $y_j$, which we can interpret as a distance, is noted as $d(x_i,y_j)$, resulting into a distance matrix ${\bf D}=\left(d(x_i,y_j)\right)$. We define $\theta(s,t)$ as the share of the commodity in $s$ that should be transported into $t$. The transport plan is the matrix $\Theta=(\theta(s,t))_{m\times n}$ and should validate commodity conservation constraints, {\em i.e.}, $\sum_{t \in \mathcal{D}}\theta(s,t)=\mu(s)$ and $\sum_{s \in \mathcal{S}}\theta(s,t)=\nu(t)$. The total transport cost of a transport plan is :
\begin{equation}
    C(\Theta, \mu, \nu)=\sum_{s \in \mathcal{S},t \in \mathcal{D}}\theta(s,t)d(s,t).
\end{equation}
The optimal transport problem search for the transport plan, denoted $\Theta^*(\mu,\nu)$, with the minimal total transport cost  $C^*(\mu, \nu)$:
\begin{align}
\Theta^*(\mu,\nu) = \arg\min_{\Theta} C(\Theta,\mu,\nu), \;C^*(\mu, \nu)=\sum_{s \in \mathcal{S},t \in \mathcal{D}}\theta^*(s,t)d(s,t)
\label{eq:optTransport},
\end{align}
This is an instance of a linear program with $|\mathcal{S}|\times|\mathcal{D}|$ variables, where the commodity initial distribution $\mu$ and the final one $\nu$, are given as well as the distance matrix ${\bf D}$. The optimal transport $\Theta^*$ always exists for a connected graph,  and can be computed exactly or approximated efficiently (see for example \cite{Cuturi}). Solving the optimal transport problem over a graph is straightforward using the simplex algorithm or interior point methods. We will later show that in practice, the cost of calculating the sources to destinations distance matrix accounts for most of the total cost of calculating the optimal transport distance and that solving the linear problem is only a small portion of the total cost. 

An important property of optimal transport is that it validates triangular inequality as stated below
\begin{lem}
    Let's suppose three measure $\mu$, $\psi$ and $\nu$ defined over the vertices of a graph $G=(V,E)$, we have the triangular inequality:
    $$
    C^*(\mu, \nu)+C^*(\nu, \psi)\ge C^*(\mu, \psi)
    $$
\end{lem}
\begin{proof}
if $\theta(\mu,\nu)$, resp. $\Omega(\nu,\psi)$, is a transport plan between measures $\mu$ and $\nu$, resp. $\nu$ and $\psi$, $\theta(\mu,\nu)\times \Omega(\nu,\psi)$ is a transport plan between $\mu$ and $\psi$. However, such a plan might not be the one optimizing the cost of transferring $\mu$ to $\psi$, leading to the triangular inequality 
\end{proof}
The optimal transport depends on two ingredients: the distance matrix $D=(d(s,t))$, and the initial and final commodity distribution.  While the distance matrix is clearly a structural property of the graph, the initial and final commodity distributions are related to what the graph represents. Therefore the curvature that uses the Wasserstein distance can integrate the structure and properties of the underlying process represented by the graph. 
\subsection{Discrete Ricci curvature}
\label{sec:ricci}
Different notions of curvature, all representing and quantifying deviation of the space from flatness, have been defined in differential geometry \cite{lafontaine}. We use here, the Ricci curvature over a manifold, which is defined by looking into the transport, along the geodesics of the manifold, of a ball of radius $\epsilon$ around a point, to a ball around a close-by point. When the geodesics related to nearby points diverge and converge (resp. converges and diverges), we have a positive (resp. negative) Ricci Curvature. When the geodesics remain parallel we have zero curvature.

While graphs are not manifolds, one can consider the neighborhood of the vertex as a ball around it, and the shortest paths as playing the role of the geodesics. In \cite{Ollivier:2009}, the transport of the neighborhood $N(x)$ of the x extremity of the edge $(x,y) \in E$, vertex $x$ to the neighborhood $N(y)$ of $y$ is analyzed. Let $\mu^{\alpha}_x$, $\nu^{\alpha}_y$ be two distributions defined over graph vertices as:
\begin{equation}
\label{eqn:dist}
\mu^{\alpha}_x(z)= \frac{(1-\alpha) w_{xz} }{\sum_{t \in N(x)}w_{xt}}\mathbbm{1}_{N(x)}(z)+\alpha \mathbbm{1}_x(z)
\;,    \nu^{\alpha}_y(z)= \frac{(1-\alpha) w_{yz}}{\sum_{t \in N(y)}w_{yt}}\mathbbm{1}_{N(y)}(z)+\alpha \mathbbm{1}_y(z)
\end{equation}
where $N(.)$ is the neighborhood set, and $\mathbbm{1}_A$ the indicator function of set $A$. One can apply the optimal transport framework defined above to derive the optimal transport plan  $\Theta^*(x,y, \alpha)$ , and  $C^*(x,y,\alpha)$. 
The $\alpha$-curvature of an edge $(x,y)$, also called Ricci idleness, is defined as: 
\begin{equation}
    \kappa(x,y,\alpha) = 1 - \frac{C^*(x,y,\alpha)}{d(x,y)},
\end{equation}
It is notable that $\kappa(x,y,1)=0$ for all graphs and edges. In \cite{Ollivier:2009}, the curvature of an edge $(x,y)$ is defined as $\kappa(x,y,0)$, {\em i.e.}, the distribution used for calculating the optimal transport in Ollivier's definition of curvature take values only in the neighborhood of the extremities and not in the center (the edge extremity). This definition has been used frequently in the literature. For example \cite{7218668} and \cite{10.1145/3489048.3522645} use the $\frac{1}{2}$-curvature, and \cite{Sia2019OllivierRicci} and\cite{Ni2019Community} use the 0-curvature (the Ollivier Ricci definition).

In this paper, we will use a slightly different definition of curvature, the one proposed by Lin, Lau, and Yau (LLY curvature) \cite{LLY}. For the edge (x,y) the LLY curvature is defined as: 
\begin{equation}
\label{eq:LLYRicci}
    \kappa(x,y)=\lim_{\alpha \rightarrow 1}\frac{\kappa(x,y,\alpha)}{1-\alpha}
\end{equation}
By letting the distribution to concentrate in the extremities ($\lim_{\alpha \rightarrow 1}$), this definition is closer to the spirit of the initial Ricci curvature definition that looks at an infinitesimal neighborhood around a point. Moreover, the calculation of LLY curvature involves calculating a single $\alpha$-curvature for $\alpha \approx 1$. This results from the fact that $ \kappa(x,y,1)=0$, and the result proven in \cite{doi:10.1137/17M1134469} that $\kappa(x,y,\alpha)$ as a function of $\alpha$ is piece-wise linear with at most 3 linear parts, {\em i.e.},  the LLY curvature is the slope of the last linear part. It is notable that for a large majority of configurations the final values of $\kappa(x,y)$ and $\kappa(x,y,0)$ are equal, reconciliating the initial Ollivier's curvature definition and the LLY one. Last but not least, the LLY definition of curvature accepts a remarkable  limitless Kantorovitch dual formulation \cite{MUNCH2019106759} as:
\begin{equation}
\label{eq:Kanto}
\kappa(x,y)=\inf_{\substack{\nabla_{yx}f=1\\ f \in Lip(1)}} \nabla_{xy}\Delta f
\end{equation}
where the $\nabla_{yx}f=\frac{f(y)-f(x)}{d(x,y)}$ is the gradient of $f$ over the edge $(x,y)$, and $\Delta f$ the Laplacian function of $f$  defined as :
\begin{equation*}
\Delta f(s)=\frac{\sum_{t \in N(s)}{w_{ts}\left( f(t)-f(s)\right)}}{\sum_{t \in N(s)}{w_{ts}}}
\end{equation*}
The optimization is done over all Lipschitz(1) real-valued functions defined over the graph vertices $f:V\rightarrow \mathbbm{R}$ , {\em i.e.}, functions such that  $|f(s)-f(t)| \leq d(s,t)$, for all edges $(s,t)$, where $f(y)-f(x)=d(x,y)$. 

As the limitless Kantorovitch optimization has $|V|$ variables (the values of $f$ at each vertex), that is generally larger than $|\mathcal{S}|\times|\mathcal{D}|$ as in the formulation in Eq. \ref{eq:optTransport}, this formulation is generally not useful for calculating numerically the curvature in practice. However, it is fundamental for proving the convergence of the discrete Ricci flow (sec. \ref{sec:convergence}), as the classical Ollivier's definition is not amenable to a convergence proof. 

We will use throughout this paper, the LLY definition of curvature for graph edges, named Discrete Ricci Curvature (DRC). It is a notable difference with what has been done in the literature where the $0$-curvature and the $\frac{1}{2}$-curvature has been used to apply Ricci Curvature to practical problems \cite{10.1145/3489048.3522645, osti_10153644,Ni2019Community,Infocom}.  It is important to consider that the distribution defined in Eq. \ref{eqn:dist} makes it possible to integrate into the curvature both the structural properties of the graph and intrinsic properties of the underlying process through the edge's weight. By setting the weights to 1, the obtained curvature will only represent the graph structure. For all graphs, we have $-2\le \kappa(x,y)\le 2$. 

Let us look at three extreme cases to develop an intuition on what DRC represents.
\begin{itemize}
\item For a complete graph with $N$ vertex, all neighbors of $x$ are neighbors of $y$. So the optimal transport plan $\Theta^*$ has nothing to transport beyond a mass $\alpha-\frac{1 - \alpha}{N-1}$ from $x$ to $y$, resulting into $\kappa(x,y)=\frac{N}{N-1}>1$. The complete graph is the only graph with edge's distance equal to 1 and a constant Ricci curvature greater than 1 \cite{LLY}.
\item For a two-star network, with $N/2$ nodes each (one in the center and $\frac{N}{2}-1$ in the branches), connected by the edge (x,y), $\kappa(x,y) = -2$.
\item For a two-star network similar to the previous case, with centers connected to each other, with each star's branch being also directly connected to another branch of the other star, each branch can send its mass directly to its counterpart with a distance 1, resulting in $\kappa(x,y) = 0$.
\end{itemize}
We rarely see situations precisely like the ones described above, in real graphs. However, one can expect to see DRC larger than 0 in cases where the network is very densely connected, the DRC being negative, when there is a bottleneck, and the DRC being close to 0 when parallel paths exist. This motivates the use of DRC to evaluate the graph structure. 
\subsection{Ricci Flow}
\label{sec:ricciflow}
The Ricci flow over a smooth manifold $M$ with a Riemannian metric $g(t)$, introduced in \cite{Hamilton88}, is a non-linear Partial Differential Equation (PDE) that acts on the manifold through the metric $g(t)$:
\begin{equation}
\label{eq:ricciflow}
\frac{\partial g(t)}{\partial t}=-2\kappa^c(g(t))   
\end{equation}
The Ricci flow acts by deforming the distance metric $g(t)$ proportionally to its Ricci curvature, $\kappa^c(g(t))$, that depends itself on $g(t)$, {\em i.e.} it reduces the distance metric in the regions with positive curvature and increases in regions with negative curvature. This results into making the curvature of the manifold constant, (e.g, flat for Euclidean spaces,) by changing $\kappa^c(g(t))$. 

The Perelman-Poincar{\'e} theorem \cite{cao2006hamiltonperelmans} proves that when the Ricci flow is applied to a 3-manifold it converges to a composition of pieces each having one of 8 possible geometries. Three of these geometries are among the regular spaces (Einstein spaces): Euclidean, spherical, and hyperbolic, the 5 other geometric structures are a combination of regular space (through cartesian products). In the process of convergence, the Ricci flow might get into a countable number of singularities, {\em i.e.}, points where curvature diverges. These singularities are dealt with surgery, {\em i.e.} the manifold is cut along the singularities and split into several pieces over which the Ricci flow is pursued. The emergence of five geometric structures in addition to the regular manifold is coming from specific surgery cases. In a nutshell, among 3-manifolds, a large number is projected in a homeomorphic way into a regular space, {\em i.e.}, a  bijective and continuous function between the two spaces exists that has a continuous inverse function, through the Ricci flow, resulting in the two space being topologically equivalent. But for some other cases we need to apply surgery, and more complex geometries are possible.

In this paper, we are interested in the geometric representation of finite connected graphs. Similarly to the extension of Ricci curvature to the Discrete Ricci Curvature, we look for an extension of the Ricci flow to a discrete Ricci flow. We will investigate the existence of singularity and the need for surgery for discrete Ricci flows. It is also of major interest to know if regular constant curvature spaces are sufficient for embedding graphs or do we need more complex product space, like the 5 additional spaces needed for a 3-manifold. We will investigate these questions in forthcoming sections. 
\section{Discrete Ricci flow}
In \cite{Ollivier:2009}, an extension of the Ricci flow to graphs is proposed, it discretizes the PDE in Eq. \ref{eq:ricciflow}, by a discrete recurrence equation on the Ollivier version of DRC, $\kappa^k(x,y,0)$:
\begin{equation}
\label{eq:Olldiscricciflow}
d^{k+1}(x,y)=d^k(x,y)(1-\kappa^k(x,y,0))  
\end{equation}
with $d^k(x,y)$ and $\kappa^k(x,y,0)$ being the distance and the Ollivier curvature of an edge $(x,y)$ at the $k^{\textrm{th}}$ iteration. This definition has been used in \cite{Ni2019Community, Weber_1, 10.1007/978-3-030-04414-5_32, osti_10153644} to apply Ricci flow to different applications. However, there is no proof that the recurrence in Eq. \ref{eq:Olldiscricciflow} converges. In this paper, we will use the Discrete Ricci Curvature defined in Section \ref{sec:ricci} and a slightly different version of recurrence:
\begin{equation}
\label{eq:Olldiscricciflow2}
d^{k+1}(x,y)=d^k(x,y)\left(1-\frac{\kappa^k(x,y)}{2}\right)   
\end{equation}
Contrary to the recurrence given in Eq. \ref{eq:Olldiscricciflow}, we will prove that this recurrence converges to a fixed point(the distances of all the edges are Constance) where all edges of the graph have the same DRC:
\begin{equation}
     \forall (x,y) \in E, \;\lim_{k \rightarrow \infty}{\kappa^k(x,y)}=\kappa^*.
\end{equation} 
At each iteration $k$, the distance  $d(x,y)$ is stretched when the DRC is negative and squeezed when it is positive. Notably, the DRC is distance-scale- free, {\em i.e.}, if all distances over the graph are scaled by $\beta>0$, $d(x,y)\rightarrow \beta d(x,y)$, no edges' DRC change. Therefore for a constant DRC, the distances might still change. If  $\kappa^*>0$, resp. $\kappa^*<0$, the distance will continue to shrink to 0, resp. inflate with a constant rate.  We can leverage DRC being scale-free, to stabilize the distances by re-scaling them after each iteration to ensure that the average edge distances remain equal to 1. This results in the below recurrence equation that will be defined as the Discrete Ricci Flow:
\begin{equation}
\label{eq:discricciflow}
d^{k+1}(x,y)=\hat{d}^k(x,y)\left(1-\frac{\kappa^k(x,y)}{2}\right), 
\;
\hat{d}^{k+1}(x,y) =\frac{|E|d^{k+1}(x,y)}{\sum_{(s,t) \in E} d^{k+1}(s,t)}
\end{equation} 
We will assume here a state vector, $S^k$, of dimension $|E|$ containing the distances of all graph edges at step $k$,  ${\bf d}^k=(\hat{d}^k(e))_{e \in E}$. The state vector belongs to the simplex space, $\mathcal{S}=\left\{{\bf d}=(d_{e_1}, \ldots, d_{e_{|E|}}) |\sum_{e \in E}d_e = 1 \right\}$. The curvatures $\kappa^k(e)$ are directly derived using the state vector values $\hat{d}^{k}$. The "Discrete Ricci Flow" $\Gamma$ (DRF), which integrates equations relative to all edges is defined as the map :
\begin{equation}
\label{eq:gamma}
\hat{{\bf d}}^{k+1}=\Gamma(\hat{{\bf d}}^k)
\end{equation}
The following lemma shows that the discrete Ricci flow validates the triangular inequality:
\begin{lem}
When the metric $\hat{{\bf d}}^k$ at step $k$ validates triangular inequality, $\hat{{\bf d}}^{k+1}$ the metric after application of DRF validates triangular inequality too.
\end{lem}
\begin{proof}
Let's suppose that $\hat{d}^k(x,y)+\hat{d}^k(y,z)\ge \hat{d}^k(x,z)$ for all triangles $xyz$. Let's use the limitless Kantorovitch formulation of the curvature. let's $F \in Lip^{d}(1)$ be the function that attains the minimum for calculating $\kappa^k(x,z)$. We have $F(z)-F(x)=\hat{d}^k(z,x)$. We can write:
$$
\hat{d}^k(x,z)-\frac{1}{2}\kappa^k(x,z)\hat{d}^k(x,z)=\hat{d}^k(x,z)-\frac{1}{2}\left(\Delta F(x)-\Delta F(z)\right)
$$
Let $\alpha_{xu}= \frac{w_{xu}}{\sum_{u \in N(x)}}$, $\gamma_{zu}= \frac{w_{zu}}{\sum_{u \in N(z)}}$ and $\beta_{yu}= \frac{w_{yu}}{\sum_{u \in N(y)}}$, then $\sum_{u \in N(x)}\alpha_{xu}=\sum_{u \in N(z)}\gamma_{zw}=\sum_{u \in N(y)}\beta_{yu}=1$, we can write $\hat{d}^k(x,z)=\frac{1}{2}\sum_{u \in N(x)}{\alpha_{xu}\left(F(z)-F(x)\right)}+\frac{1}{2}\sum_{w \in N(z)}{\gamma_{zw}\left(F(z)-F(x)\right)}$. Therefore:
\begin{equation}
\label{eq:d1}
    \hat{d}^k(x,z)-\frac{1}{2}\kappa^k(x,z)\hat{d}^k(x,z)=\frac{1}{2}\left(\sum_{w \in N(z)}{\gamma_{zw}\left(F(w)-F(x)\right)}-\sum_{u \in N(x)}{\alpha_{xu}\left(F(u)-F(z)\right)} \right)
\end{equation}
We can use a similar approach to rewrite
\begin{equation}
\label{eq:d2}
\hat{d}^k(x,y)-\frac{1}{2}\kappa^k(x,y)\hat{d}^k(x,y)=\frac{1}{2}\left(\sum_{v \in N(y)}{\beta_{yv}\left(G(v)-G(x)\right)}-\sum_{u \in N(x)}{\alpha_{xu}\left(G(u)-G(y)\right)} \right)
\end{equation}
where $G$ is the function attaining the minimum for edge $(x,y)$,
\begin{equation}
\label{eq:d3}
\hat{d}^k(y,z)-\frac{1}{2}\kappa^k(y,z)\hat{d}^k(y,z)=\frac{1}{2}\left(\sum_{w \in N(z)}{\gamma_{zw}\left(H(w)-H(y)\right)}-\sum_{v \in N(y)}{\beta_{yv}\left(H(v)-H(z)\right)} \right)
\end{equation}
where  $H$ is the function attaining the minimum for edge $(y,z)$. It is noteworthy that $F$,$G$, and $H$ are all three in $Lip^{d}(1)$. Moreover, it is notable that the term on the right-hand side of Eq. \ref{eq:d1} does not depend on $\hat{d}(x,z)$ as the term $(F(z)-F(x))$ never appear in the summations. Similarly Eq. \ref{eq:d2} does not depend on $\hat{d}(x,y)$, and  Eq. \ref{eq:d3} does not depend on $\hat{d}(y,z)$. 

Now if we replace $G$ in Eq. \ref{eq:d2}, and $H$ in Eq. \ref{eq:d3}, by the function $F$, we are decreasing both of the right-hand sides so that:
\begin{dmath*}
\hat{d}^{k+1}(x,y)+ \hat{d}^{k+1}(y,z)\ge \frac{1}{2}\left(\sum_{v \in N(y)}{\beta_{yv}\left(F(v)-F(x)\right)}-\sum_{u \in N(x)}{\alpha_{xu}\left(F(u)-F(y)\right)} \right)+\frac{1}{2}\left(\sum_{w \in N(z)}{\gamma_{zw}\left(F(w)-F(y)\right)}-\sum_{v \in N(y)}{\beta_{yv}\left(F(v)-F(z)\right)} \right)
\end{dmath*}
By re-arranging the equations and removing the common summation term we have:
$$
\hat{d}^{k+1}(x,y)+ \hat{d}^{k+1}(y,z)\ge\frac{1}{2}\left(\left(F(z)-F(x)\right)+\sum_{w \in N(z)}{\gamma_{zw}\left(F(w)-F(y)\right)}-\sum_{u \in N(x)}{\alpha_{xu}\left(F(u)-F(y)\right)} \right)
$$
Now we can write :
$$
F(z)-F(x)=\sum_{w \in N(z)}{\gamma_{zw}\left(F(y)-F(x)\right)}-\sum_{u \in N(x)}{\alpha_{xu}\left(F(y)-F(z)\right)}
$$
By replacing this term in the above inequality we have :
$$
\hat{d}^{k+1}(x,y)+ \hat{d}^{k+1}(y,z)\ge\frac{1}{2}\left(\sum_{w \in N(z)}{\gamma_{zw}\left(F(w)-F(x)\right)}-\sum_{u \in N(x)}{\alpha_{xu}\left(F(u)-F(z)\right)} \right)=d^{k+1}(x,z)
$$
\end{proof}
In this paper, we are interested into checking if the DRF converges to a fixed point, {\em i.e.}, a configuration where the distances remain constant, $\hat{d}^{k+1}(e)=\hat{d}^k(e)$. Nonetheless, at this fixed point, if it exists, the DRC will be constant over all edges as the only way for the distance to remain stable is for the DRC of all edges to be constant so that the inflation of distance is canceled by the rescaling. 

As described in sec. \ref{sec:ricciflow}, Ricci flow over manifolds, might stumble upon singularities that will need surgery. These singularities happen when curvature becomes infinite, {\em i.e.}, when a part of the manifold becomes a ball or a cylinder, with a radius going to 0. However, the discrete curvature defined in this paper is bounded, and it cannot diverge. The case where a distance goes to 0 can be considered as problematic. This will not break the convergence of the distance vector to a stable value, as when a distance becomes 0, it remains 0. Even when the distance goes to 0, the DRC defined in Eq. \ref{eq:LLYRicci} or in Eq. \ref{eq:Kanto}, is still well-defined. Therefore, in the process of deriving the DRF, we formally do not need surgery. This is an important difference between our context with classical manifold setting.

This also goes on contrary, to the approach advocated in \cite{Sia2019OllivierRicci, osti_10153644}, which applies surgery when distances go beyond a threshold. While one can still decide to sever some links to implement the clustering, this severing should not be motivated by surgery.

In the next section, we will prove that the recurrence in Eq. \ref{eq:discricciflow} converges to a unique fixed point.

\section{Convergence and Embedding}
\label{sec:convergence&embedding}

\subsection{Convergence}
In the forthcoming, we will use the  $L_\infty$ metric over the simplex space defined as $L_{\infty}({\bf d}, {\bf d'})=\max_{e \in E} |d'(e)-d(e)|$. The simplex space along with the $L_{\infty}$ metric define a complete metric space.




\begin{lem}\label{lemma:1}
Let $G=(V, E,w)$ be a weighted graph, and $e=(x,y)\in E$ an edge in the graph. For any two distance vectors ${\bf d}$ and ${\bf d}'$, such that ${\bf d}$ and ${\bf d}'$ differ only on edge $(x,y)$ and $d(x,y) \leq d'(x,y)$, and for any function $f\in Lip(1)$, such that $f(y)-f(x) = d(x,y)$ there exists a function $f' \in Lip(1)$, such that $f'(y)-f'(x) = d'(x,y)$, and $0\leq f'(z)-f(z)  \leq d'(x,y) - d(x,y)$ for all $z$ in $V$.
\end{lem}
\begin{proof}
The lemma states that for a function $f \in Lip(1)$, that validates the distance constraint on the edge $(x,y)$, $f(y)-f(x)=d(x,y)$, there exists an approximation $f' \in Lip(1)$, which validates the distance constraint $f'(y)-f'(x)=d'(x,y)$.

We know that $f$ is in $Lip(1)$, so $|f(s) -f(t)| \leq d(s,t)$ for all $(s,t)$ in $E$. 
Now, let us define the function $\mathcal{F}$ as:
$$
\mathcal{F}(z) = \left\{
    \begin{array}{ll}
        \ f(x) + d'(x,y) & \mbox{if}\; z = y \\
         f(z) & \mbox{otherwise}
    \end{array}
\right.
$$
By definition of $\mathcal{F}$, we have $\mathcal{F}(y) - \mathcal{F}(x) = d'(x,y)$, so $\mathcal{F}$ validate the distance constraint on edge $(x,y)$. Moreover,  $\mathcal{F}(y) - f(y) = f(x)-f(y)+ d'(x,y) = d'(x,y) - (f(y) - f(x)) = d'(x,y) - d(x,y)$
, and obviously $\mathcal{F}(z)-f(z)=0$ for $z\neq y$ and therefore $\forall z \in V,\; 0\leq \mathcal{F}(z)-f(z) \leq d'(x,y) - d(x,y)$.

If the function $\mathcal{F}$ is not in $Lip(1)$ there exists a set of $n$ neighbors of vertex $y$  $\{z_1, \ldots, z_n\}$, indexed such that $\mathcal{F}(z_1)\le \mathcal{F}(z_2)\le \ldots\le \mathcal{F}(z_n)<\mathcal{F}(y)$, and $|\mathcal{F}(z_i)-\mathcal{F}(y)|>d'(z_i,y)$. Let $Z=\{z_1,\ldots,z_n,y\}$. It is notable that the function $\mathcal{F}$ is only different from $f$ on vertex $y$, therefore, it inherits from $f$ the property $Lip(1)$ for all other edges beyond those with extremities in $Z$. 

We know $\mathcal{F}(z_i)-\mathcal{F}(y)=\mathcal{F}(z_i)-\mathcal{F}(x)+\mathcal{F}(x)-\mathcal{F}(y) \le d'(x,z_i)-d'(x,y)\le d'(z_i,y)$ by the triangular inequality. So we have
$\mathcal{F}(z_i)-\mathcal{F}(y)<-d'(z_i,y)$. 

Now, let us define a function $\mathcal{F}'$ from $\mathcal{F}$, as $\mathcal{F}'(v)=\mathcal{F}(v)$ for all vertices beyond $z_i$ where $\mathcal{F}'(z_j)=\mathcal{F}(z_j)+\delta_j$. We select $\delta_j>0$ such that $\mathcal{F}'(z_j)=\mathcal{F}(z_j)+\delta_j>\mathcal{F}(y)-d'(y,z_j)$, resulting in $\mathcal{F}'$ being $Lip(1)$ on all vertices $z_j$. The additional value $\delta_j$ is upper bound by $d'(x,y)-d(x,y)$. It is also lower-bounded as $\delta_j \geq \mathcal{F}(y)-\mathcal{F}(z_j)-d'(z_j,y)$, as this is the minimal value to add to achieve the $Lip(1)$ property on edge $(z_j,y)$. 

However, changing the values of $\mathcal{F}(z_j)$ might impact $Lip(1)$ property for any edge connected to $z_j$. We know that $\mathcal{F}(y)-\mathcal{F}(w) \le d'(w,y)$, for $w \notin {Z}$, 
Let's consider a vertex $w \in N(y)\backslash Z$ connected to $z_j$. As $\mathcal{F}$ is $Lip(1)$, there are three cases: $\mathcal{F}(z_j)-\mathcal{F}(w)=d'(z_j,w)$, or $\mathcal{F}(w)-\mathcal{F}(z_j)=d'(z_j,w)$  $|\mathcal{F}(z_j)-\mathcal{F}(w)|< d'(z_j, w)$. In the first case,
$\mathcal{F}(w)-\mathcal{F}(y)=\mathcal{F}(z_j)-d'(z_j,w)-\mathcal{F}(y)<-d'(z_j,y)-d'(z_j,w)<-d'(w,y)$, but this results in a contradiction as this means that $\mathcal{F}$ is not $Lip(1)$ over $(w,y)$, and we assumed that $w \notin Z$. For the second case, $\mathcal{F}(w)-\mathcal{F}(z_j)=d'(z_j,w)$, we have $\mathcal{F}(w)-\mathcal{F}'(z_j)=\mathcal{F}(w)-\mathcal{F}(z_j)-\delta_j<d'(z_j,w)$, so that we do not need to add anything to $\mathcal{F}(w)$ to make $\mathcal{F}' \in Lip(1)$ for the edge $(w,z_j)$. This is similar for the third case, where $|\mathcal{F}(z_j)-\mathcal{F}(w)|< d'(z_j, w)$. Therefore, setting $\mathcal{F}'(w)=\mathcal{F}(w)$ the function $\mathcal{F}'$ remains in $Lip(1)$ on all edges $(w,z_j)$ with $w \in N(y)-Z$. 

The last part that we have to take care of in the neighborhood of $y$ is to ensure that the property $Lip(1)$ is valid among all edges $(z_i,z_j)$ where both extremities are in $Z$. For this purpose, we will let $0<\delta_j-\delta_i<\mathcal{F}(z_i)-\mathcal{F}(z_j)+d'(z_i,z_j)$. We can ensure this by selecting the highest value of $\delta_i$ in decreasing order, {\em i.e.}, beginning with $\delta_t$, $\delta_{t-1}$, up to $\delta_1$. This guarantees $\mathcal{F}'$ is in $Lip(1)$ for all edges with both extremities in $Z$. 

We have therefore proven that there exists a set of values $\{ \delta_i\}$, such that the function $\mathcal{F}'$ defined above is $Lip(1)$ for all edges of the graph. Moreover, the function $\mathcal{F}'$ validates constraints over $(x,y)$, and the approximation constraint is validated on all vertices as $\mathcal{F}'(z)-f(z)=0$ for all $z \notin Z$, and for $z_i \in Z$ we have $\mathcal{F}'(z_i)-f(z_i)=\delta_i$, where $d'(x,y)-d(x,y)\le \delta_i \le 0$, {\em i.e.}, $0 \le \mathcal{F}'(z_i)-f(z_i)\le d'(x,y)-d(x,y)$ so that $\mathcal{F}'$ validates the approximation constraint.
\end{proof}

\begin{lem}\label{lemma:2}
Let $G=(V, E,w)$ be a weighted graph, and $e=(x,y)\in E$ an edge in the graph. For any two distance vectors ${\bf d}$ and ${\bf d}'$, such that ${\bf d}$ and ${\bf d}'$ differ only on edge $(x,y)$ and $d(x,y) \leq d'(x,y)$, and for any function $g\in Lip(1)$, such that $g(y)-g(x) = d'(x,y)$, there exists a function $g' \in Lip(1)$, such that $g'(y)-g'(x) = d(x,y)$, and $d(x,y) - d'(x,y)\leq g'(z)-g(z) \leq 0$ for all $z$ in $V$ .
\end{lem}
\begin{proof}
The proof of this lemma follows a similar logic as the previous one, and can be found in Appendix \ref{appendix:proof}
\end{proof}

\newtheorem{thm}[lem]{Theorem}
\begin{thm}
   Let $G=(V, E,w)$ be a weighted graph. For any edge $(x,y) \in E$, the discrete Ricci flow, defined in Eq. \ref{eq:gamma} as:
   \begin{equation}
       \hat{{\bf d}}^{k+1}=\Gamma(\hat{{\bf d}}^k),
   \end{equation}
    converge to a unique fixed point where all edges have the same discrete Ricci curvature, $\kappa^{\hat{{\bf d'}}}(x,y)=\kappa^{\hat{\bf d'}}_{\infty}$, and the distance of each edge is stable, $\hat{d}^{k+1}(x,y)=\hat{d}^k(x,y)$. 
\end{thm}
\begin{proof}
To prove the convergence of the recurrence, we will use the Banach fixed point theorem, which states that a contraction mapping defined over a complete metric space admits a unique fixed point. We therefore have to prove that the map $\Gamma(.)$ is contractive over the simplex space with the metric $L_{\infty}$.
Therefore, we need to prove that $\Gamma$ is contractive, {\em that is,}, $\epsilon =L_{\infty}(\hat{{\bf d}}, \hat{{\bf d'}}) \geq L_{\infty}(\Gamma(\hat{{\bf d}}), \Gamma(\hat{{\bf d'}}))$. 

For proving the theorem over two general distance vectors ${\bf d}$ and ${\bf d'}$, we consider a chain of intermediate steps $\hat{{\bf d}}=\tilde{{\bf d}_0}\rightarrow \tilde{{\bf d}_1} \rightarrow \ldots \rightarrow \tilde{{\bf d}_n}= {\bf d'}$, where $\tilde{{\bf d}_i}$ is a vector that differ with $\tilde{{\bf d}_{i-1}}$ in a single coordinate, $j=\arg \max |\tilde{d}_{i-1}(k)-\tilde{d_{i}}(k)|$. We will first prove that the function $\Gamma'()$, the $\Gamma$ mapping without normalization, is contractive over each step $\tilde{{\bf d}}_{i-1}\rightarrow \tilde{{\bf d}}_i$, {\em, that is,}, $L_{\infty}\left(\tilde{{\bf d}}_{i-1}, \tilde{{\bf d}}_{i}\right) \geq L_{\infty}(\left(\Gamma'(\tilde{{\bf d}}_{i-1}), \Gamma'(\tilde{{\bf d}}_{i})\right)$.  

In order to simplify the notation in the coming, we will denote $\tilde{\bf d}_{i-1}$ as $\tilde{\bf d}$, $\tilde{\bf d}_{i}$ as $\tilde{\bf d}'$. Let $e = (x,y)$ be the edge of the graph, where $\tilde{{\bf d}}$ and $\tilde{{\bf d}}$ differ. Without loss of generality, we will assume that $\tilde{d}'(x,y) \geq \tilde{d}(x,y)$,  $0 \leq \tilde{d}'(x,y)-\tilde{d}(x,y) = \epsilon_{0}$. 
 We have the following:
\begin{equation*}
|\Gamma(\tilde{d}'(x,y))-\Gamma(\tilde{d}(x,y))|=\left|\left(\tilde{d}'(x,y)-\frac{\kappa^{{\bf \tilde{d}'}}(x,y)}{2}\tilde{d}'(x,y)\right) -\left(\tilde{d}(x,y)-\frac{\kappa^{{\bf \tilde{d}}}(x,y)}{2}\tilde{d}(x,y)\right) \right|
\end{equation*}
{\bf First case}: we suppose that $\Gamma(\tilde{d}'(x,y)) \geq \Gamma(\tilde{d}(x,y)$. Therefore:
\begin{equation*}
\label{eq:1devmapping}
\begin{aligned}
 \Gamma(\tilde{d}'(x,y))- \Gamma(\tilde{d}(x,y))& = \left(\tilde{d}'(x,y)-\frac{\kappa^{{\bf \tilde{d}'}}(x,y)}{2}\tilde{d}'(x,y)\right) - \left(\tilde{d}(x,y)-\frac{\kappa^{{\bf \tilde{d}}}(x,y)}{2}\tilde{d}(x,y)\right)  \\
& = (\tilde{d}'(x,y) - \tilde{d}(x,y)) - \frac{1}{2}\left(\tilde{d}'(x,y)\kappa^{{\bf \tilde{d}'}}(x,y) -  \tilde{d}(x,y)\kappa^{{\bf \tilde{d}}}(x,y)\right) \\
& = \epsilon_{0} - \frac{1}{2}\left(\tilde{d}'(x,y)\kappa^{{\bf \tilde{d}'}}(x,y) - \tilde{d}(x,y)\kappa^{{\bf \tilde{d}}}(x,y)\right) \\
& = \epsilon_{0} - \frac{1}{2}\left(\tilde{d}'(x,y)\inf_{\substack{\nabla_{yx}f=1\\ f \in Lip^{{\bf d}'}(1)}} \nabla_{xy}\Delta f
- 
\tilde{d}(x,y)\inf_{\substack{\nabla_{yx}g=1\\ g \in Lip^{\bf d}(1)}} \nabla_{xy}\Delta g\right)\\
\end{aligned}
\end{equation*}
Let us set $I = \tilde{d}'(x,y)\inf_{\substack{\nabla_{yx}f=1\\ f \in Lip^{{\bf d}'}(1)}} \nabla_{xy}\Delta f \hspace{3mm}-\hspace{3mm} \tilde{d}(x,y)\inf_{\substack{\nabla_{yx}g=1\\ g \in Lip^{\bf d}(1)}} \nabla_{xy}\Delta g.$
%
Let $F \in Lip(1)$ be the function that reaches the minimum in the first term above, {\em, that is,} $F(y)-F(x)=\tilde{d}'(x,y)$, and $\inf_{\substack{\nabla_{yx}f=1\\ f \in Lip^{\bf d}(1)}} \nabla_{xy}\Delta f= \Delta F(x) - \Delta F(y).$
Through Lemma \ref{lemma:2}, there exists a function $f' \in Lip(1)(1)$ such that $f'(y) - f'(x) = \tilde{d}(x,y)$ and $f'(z) - F(z) \leq \tilde{d}'(x,y) - \tilde{d}(x,y)$ for all $z$ in $V$. This function validates therefore constraints relative to the second minimization in the above formula. Putting $f'$ in place of $g$ in $I$ we will have:
\begin{equation*}
\begin{aligned}
I & \geq (\Delta F(x) - \Delta F(y)) - (\Delta f'(x) - \Delta f'(y))  \\
& \geq (\Delta F(x) - \Delta f'(x)) - (\Delta F(y) - \Delta f'(y)) 
\end{aligned}
\end{equation*}
We can replace the Laplacian terms with their expansions. with  $\alpha_{xz}=\frac{w_{xz}}{\sum_{t \in N(x)}wxt}$ and $\beta_{yv}=\frac{w_{yv}}{\sum_{t \in N(y)}w_{yt}}$.
\begin{equation*}
\begin{aligned}
I & \geq \left(\sum_{z\in N(x)}{\alpha_{xz}(F(z) - F(x))} - \sum_{z\in N(x)}{\alpha_{xz}(f'(z) - f'(x))}\right)- \left(\sum_{v\in N(y)}{\beta_{yv}(F(v) - F(y))} - \sum_{v\in N(y)}{\beta_{yv}(f'(v) - f'(y))}\right)\\
& \geq \left(\sum_{z\in N(x)}{\alpha_{xz}(F(z) - f'(z))} - \sum_{z\in N(x)}{\alpha_{xz}(F(x) - f'(x))}\right) - \left(\sum_{v\in N(y)}{\beta_{yv}(F(v) - f'(v))} - \sum_{v\in N(y)}{\beta_{yv}(F(y) - f'(y))}\right)\\
&\geq \left(\sum_{z\in N(x)}{\alpha_{xz}(F(z) - f'(z))} - (F(x) - f'(x))\right)-\left(\sum_{v\in N(y)}{\beta_{yv}(F(v) - f'(v))} - (F(y) - f'(y))\right)\\
&\geq \left(\sum_{z\in N(x)}{\alpha_{xz}(F(z) - f'(z))}
- \sum_{v\in N(y)}\beta_{yv}(F(v) - f'(z))\right) + (F(y) - F(x))-(f'(y)- f'(x))\\
&\geq \left(\sum_{z\in N(x)}{\alpha_{xz}(F(z) - f'(z))}
- \sum_{v\in N(y)}\beta_{yv}(F(v) - f'(v))\right) + (\tilde{d}'(x,y) - \tilde{d}(x,y))\\
&\geq \left(\sum_{z\in N(x)}{\alpha_{xz}(F(z) - f'(z))}
- \sum_{v\in N(y)}\beta_{yv}(F(v) - f'(v))\right) + \epsilon_0\\
\end{aligned}    
\end{equation*}
Now we can use the lemma approximation bounds.  We know from lemma \ref{lemma:2}, $0 \geq F(z)-f'(z)\geq -\epsilon_0$ By replacing these values in the above relations we have 
\begin{equation*}
        I \ge  \left(\sum_{z\in N(x)}{\alpha_{xz}(-\epsilon_0}) - \sum_{v\in N(y)}\beta_{yv}(0)\right) + \epsilon_0 \ge 0
\end{equation*}
This results into 
\begin{equation*}
        0 \leq \Gamma(\tilde{d}'(x,y)) - \Gamma(\tilde{d}(x,y))\leq \epsilon_{0} - \frac{I}{2}\le \epsilon_0\\
\end{equation*}
This proves that when $\Gamma(\tilde{d}'(x,y)) \geq \Gamma(\tilde{d}(x,y)$, $\Gamma$ is a contractive map. 

{\bf Second case}:  $\Gamma(\tilde{d}'(x,y)) < \Gamma(\tilde{d}(x,y)$. We follow the same approach as in the first case. 
$$0 < \Gamma(\tilde{d}(x,y)) - \Gamma(\tilde{d}'(x,y))= -\epsilon_0+\frac{I}{2}$$
where $I$ is as defined for the first case. 
Let $G$ be the function that achieves the infimum in $\inf_{\substack{\nabla_{yx}g=1\\ g \in Lip(1)}} \nabla_{xy}\Delta g = \Delta G(x) - \Delta G(y)$ and $G(y) - G(x) = \tilde{d}(x,y)$.
Using the lemma \ref{lemma:1}, there exists a function $g'$ in $Lip(1)$ such that $g'(y) - g'(x) = \tilde{d}'(x,y)$ and $0\leq G(z) - g'(z) \leq \tilde{d}'(x,y) - \tilde{d}(x,y)$ for all $z$ in $V$. If we replace $g'$ in the first term of $I$ and $G$ in the second term we have 
\begin{equation*}
I\leq \left(\Delta g'(x) - \Delta g'(y)) - (\Delta G(x) - \Delta G(y))\right).
\end{equation*}
 By applying the same set of operations that in case 1 on the term I we have :
\begin{equation*}
I \le  \left(\sum_{z\in N(x)}\alpha_{xz}(G(z) - g'(z))
- \sum_{v\in N(y)}\beta_{yv}(G(v) - G'(v))\right)+\epsilon_0
\end{equation*}
By using the approximation bound in lemma\ref{lemma:1}, the functions $g'$ and $G$ satisfy, $0\leq G(z) - g'(z) \leq \epsilon_0$. We can therefore upper bound $I$ by:
$$
I \le \left(\sum_{z\in N(x)}\alpha_{xz}(\epsilon_0)
- \sum_{v\in N(y)}\beta_{yv}(0)\right) +\epsilon_0=2\epsilon_0
$$ 
By replacing the above upper bound on $I$  in the term for $\Gamma'(\tilde{d}(x,y)) - \Gamma'(\tilde{d}'(x,y))$, we have $0 < \Gamma'(\tilde{d}(x,y)) - \Gamma'(\tilde{d}'(x,y))\le 0$ resulting into a contradiction. Therefore the second case cannot happens. 

We have already proven that for the first case the map $\Gamma'$ is contractive. That proves that the function $\Gamma'()$ is contractive at each step $\tilde{{\bf d}}_{i-1}\rightarrow \tilde{{\bf d}}_i$.  Now, if we apply the rescaling to the final value ${\bf \tilde{d}}$, we should evaluate the value $ L_{\infty}\left(\Gamma(\hat{{\bf d}})- \Gamma(\hat{\bf d}')\right)$. When comparing $\sum_{(s,t) \in E} \Gamma'(d(s,t))$ and $\sum_{(s,t) \in E} \Gamma'(d'(s,t))$, there are two cases in which it can happen: $\sum_{(s,t) \in E} \Gamma'(d(s,t)) \ge \sum_{(s,t) \in E} \Gamma'(d'(s,t))$ or $\sum_{(s,t) \in E} \Gamma'(d(s,t)) < \sum_{(s,t) \in E} \Gamma'(d'(s,t))$. Let's assume that we are in the first case. We can state for any edge, $(u,v)$, that :
$$
\frac{|\Gamma'(d(u,v))-\Gamma'(d'(u,v))|}{\sum_{(s,t) \in E} \Gamma'(d'(s,t))} \le
\left|\frac{\Gamma'(d(u,v))}{\sum_{(s,t) \in E} \Gamma'(d(s,t))}- \frac{\Gamma'(d'(u,v))}{\sum_{(s,t) \in E} \Gamma'(d'(s,t))}\right| \le \frac{|\Gamma'(d(u,v))-\Gamma'(d'(u,v))|}{\sum_{(s,t) \in E} \Gamma'(d(s,t))}
$$
However, we proved before that $\Gamma'()$ is contractive,, {\em i.e.}, $|\Gamma'(d(u,v))-\Gamma'(d'(u,v))|<\epsilon$, so that $ L_{\infty}\left(\Gamma(\hat{{\bf d}})- \Gamma(\hat{\bf d}')\right) <\epsilon$.

That proves that $\Gamma$ is a contractive mapping over the simplex $S$ and therefore $\Gamma$ converge toward a single fixed point thanks to Banach's fixed point theorem. That concludes the proof; {\em i.e.}, the curvature and the distance vector both converge to a fixed point.
\end{proof}

\subsection{Embedding of the graph}
\label{sec:embedding}
In this section, we will leverage the convergence of the DRF, and show its implications for graph embedding. The DRF provides, after the convergence of the Ricci flow, the asymptotic curvature, $\kappa_{\infty}$, that is, the curvature of the embedding space of homogeneous constant curvature. However, for all analyzed graphs, we obtain a negative curvature ($\kappa_{\infty}<0$),{\em i.e.},the embedding space is  $\mathbb{H}^K(\kappa_{\infty})$, the $k$-dimensional hyperbolic plane with curvature $\kappa_{\infty}$. It is noteworthy that while DRF gives the curvature $\kappa_{\infty}$, it does not give the dimension $k$. There is a rich literature on hyperbolic embedding of graphs that is motivated by application {\em e.g.}, capture of graph hierarchy information \cite{Chamberlain2017NeuralEO, NIPS2017_59dfa2df}, application to routing in computer networks \cite{4215803}, or to structural analysis of social networks \cite{10.1145/2582112.2582139}. 

There are also more theoretical reasons for making hyperbolic planes suitable for embedding graphs. First, it is known that any tree can be embedded isometrically (without any distortion), through the Sarkar’s Construction \cite{10.1007/978-3-642-25878-7_34}, in the hyperbolic plan $\mathbbm{H}$.However, graph are not trees in general. Nevertheless, it has been shown that any configuration of points in $\mathbb{H}^K$ can be approximated at large scale, or better said, when we look at the hyperbolic space from 'infinitely far away', by a tree \cite{HAMANN_2018}. The larger the scale, the better the approximation will be.  

Before formalizing the above approximation and proposing an embedding algorithm, let us differentiate our setting from the classical one. In the classical setting, we have a graph $G=(V,E,W)$ and we would like to embed it into a manifold. There are empirical reason to assume that an hyperbolic embedding into the $\mathbb{H}$ is a suitable one, therefore we look for the best hyperbolic embedding, {\em i.e.}, $f:V\rightarrow \mathbbm{H}^K$, that will minimize a distortion criterion that will depend on the difference between $w(v_i,v_j)$ and $d(f(v_i), f(v_j))$ for all $(v_i,v_j) \in E$. This problem, that is well-formalized in \cite{cvetkovski2016multidimensional, pmlr-v80-sala18a}, has attracted a large literature. Our setting is relatively different as we already know that the graph is embedded in $\mathbbm{H}^K(\kappa)$ with a know curvature $\kappa$, and we have already assigned a to each edge $(v_i,v_j)$ a distance metric $d(v_i,v_j)$ in $\mathbbm{H}^K$. However, as the real dimension $K$ of the space is unknown, we are looking for an approximate projection in a lower dimensional hyperbolic plan. As explained before, trees are perfectly embedded into $\mathbbm{H}$. In his seminal work on hyperbolic space, Gromov has defined the concept of $\delta$-hyperbolicity \cite{Gromov}, that evaluate the difference of a given graph with a tree, that will be a $0$-hyperbolic structure. Intuitively, the $\delta$-hyperbolicity measure how different can the sum of the lengths of two edges of a triangle from the lengths of the third edge. While for a graph embedded into an euclidean space, $\delta$ can be unbounded, for a graph embedded into $\mathbbm{H}^K(\kappa)$, $\delta < \log\left(1+\sqrt{\frac{2}{\kappa}}\right)$, in other term graph embedded into $\mathbbm{H}^K(\kappa)$ are more tree-like than when embedded in Euclidean plan. Now, the Ricci flow has this remarkable property that over the course of it execution

first setting, we have a graph and we look for a visualization

Let us formalize this.

More precisely, one can build over a set of points embedded into the hyperbolic plan a tree such that the paths over it are "quasi-geodetic", {\em i.e.}, and the length of the path over the tree is almost equal to the distance between the points. 

Hamann [22] showed that we can construct a rooted R-tree T inside Hk, such that every geodesic ray in Hk eventually converges to a ray of T . Thus, showing that any configuration of points at scale in Hk can be approximated by a tree. Additionally, larger the scale points can be better approximated by trees

Hyperbolic graph embedding 

The DRF provides a metric space over the graph nodes; {\em i.e.}, point to point distances between graph vertices, along with the curvature of a constant curvature space where the graph is embedded after the convergence of the DRF. However, the dimension of the embedding space is not given by the DRF.

For all graphs we analyzed, the limit curvature obtained was negative, meaning that embedding in a hyperbolic space should.

show that an isometric embedding into a regular constant curvature space, with curvature given as the asymptotic constant DRC of the graph, exists that will match precisely the distances assigned to edges. We will call this embedding the discrete flow Ricci graph embedding (dfRgE).
\begin{thm}
A weighted graph $G=(V, E, w)$ after convergence of the Continuous Ricci Flow, can be embedded isometrically, {\em i.e.} with points being at the same distance as edges distance, into a regular constant curvature manifold with curvature equal to the asymptotic constant curvature after DRF. 
\end{thm}
\begin{proof}
Let us first consider the simple case of connected graphs with 3 edges. All the 3 types of such graphs (star, triangle, and alignment) can be embedded in a 2-dimensional regular space. For example, a triangle with known edge distance $a$, $b$, and $c$, can be placed on a plane of 0 curvature space, using a compass through a high school-level scheme. We choose arbitrarily the position of the first point $C$, the second point $B$ can be anywhere on a circle of radius $a$ centered at $C$, and the point $A$ can be positioned on one of the intersection points of the circle centered in $C$ with radius $b$, and the circle centered in $B$ with radius $c$. This scheme works as long as the triangular inequalities are valid. The main ingredient that makes this method work is that the plane space binds the angles and length of the triangle together. This extends to a spherical space with positive curvature, resp. hyperbolic space with negative curvature, under the condition of using the corresponding spherical, resp. hyperbolic, circles.

We can extend this to larger graphs through induction. Let’s assume that, for any graph in the class of graphs with degree at most $K$, we can provide an embedding in a regular space of constant curvature of dimension less than $K$. Now, let us add to any one of these graphs of degree at most $K$ some vertices and edges so that we have a graph with degree at most $K+1$. In this new graph, any vertex is connected to at most $K+1$ other vertices. This vertex can be positioned in a $(K+1)$-dimension regular constant curvature space, at the intersection of $(K+1)$ hyperspheres of dimension $K$ centered on each one of the neighbors. This intersection will have at least one intersection point as long as all the $K+1$ points validate triangular inequalities. That works as in constant curvature space where the hyper-sphere is well defined, and the intersection of $K$ hyper-spheres of dimension $K$ has at least two intersection points. This is not the case for manifolds where you cannot guarantee that circles will intersect (because of lack of homogeneity and isotropy).
\end{proof}
The above proof proves the consistency of the metrics and the obtained curvature, with an embedding in a regular constant curvature manifold, and therefore its existence. However, it does not provide a construction for it, {\em, i.e.}, a way of finding coordinates of individual vertices in the space. Although it is straightforward to calculate the coordinates of the intersection of $K$ hyperspheres (using the correct hypersphere formula, which depends on the curvature value). However, the final coordinates of vertices depend on their order of embedding. A greedy algorithm could begin with any vertices, position it in the space, and derive the position of the next vertices depending on the position of the first, {\em etc.} However, downstream we might get into a configuration where previously positioned vertices are incompatible with the new distance constraints coming from new vertices. In a nutshell, while the embedding exists, finding the right ordering of vertices, to precisely characterize it, becomes an NP-hard problem. Nonetheless, this is not a big deal, as in most applications, we need the distance between vertices rather than their precise coordinates in space. This information is already provided by the DRF, and we know that these distances can be used to predict other distances, for example, using the cosine theorem.  

In fact, any metric distance over the edges can be embedded in a $N$ dimensional Euclidean space through spectral projections. But among the infinity of distance metric, the ones we derived through the DRC are compatible with the subset of embedding space represented by the graph vertices and edges being regular and of constant curvature. The main contribution of our paper is that we generate such a set of distances that are (geometrically) compatible with such constant Ricci curvature, while the set of distances generated by alternative methods is not.

The discrete Ricci flow non-linear diffusion process, we described in the paper, can be considered also as a "regularization", as the impact of highly curved edges (bottleneck) is propagated through the iteration of the discrete Ricci flow. This means that the distances assigned to edges are in fact resulting from the global graph structure rather than only a local property as other papers not using the Ricci flow are doing. 

The edge's distance metric derived from the DRF can therefore be used for various Machine Learning  (ML) tasks such as clustering or even robustness analysis. More generally, for ML tasks involving the inference of unknown distances using known distances.

Another aspect is relative to comparing the dfRge embedding with other classical Graph embedding techniques. Classical graph embedding as described in \cite{Cai2018} assumes that the data that have led to the graph are lying on a low-dimensional manifold and formalizes graph embedding as a structure-preserving dimensionality reduction problem.  The graph property to be preserved is interpreted as a pairwise node similarities matrix. Classical embedding methods try to project nodes with larger similarities close to each other. The differences between current graph embedding methods are mainly related to the approach used to calculate the pairwise node similarity matrix. In \cite{Cai2018} a taxonomy of different approaches is provided, where it is stated that the pairwise node similarity matrix can come from vertices directly through similarity metrics of the vertex attributes, {\em e.g}, Euclidean distance between attributes, or from spectral embedding approach. This last approach is  using the Laplacian matrix of the graph as the measure of similarity and generates a mapping between vertices and points in a $N$-dimensional Euclidean space ($N$ being the number of vertices in the graph). From there, dimension reduction techniques are applied to end up with points in a $K$-dimension with $K<N$. Muti-Dimensional Scaling (MDS) method \cite{NIPS1994_1587965f}, for example, uses the Euclidean distance between points representing vertices and reduces the dimensionality through an analysis of the spectrum of the point-wise distance matrix. Isomap method \cite{doi:10.1126/science.295.5552.7a}
considers the neighborhood of nodes by considering the graph structure and using the shortest path distances through the connected nodes to calculate the point-wise distance and thereafter applies to it the same dimensionality reduction approach as MDS. 
These two methods can be considered as defining the embedding using only the structural properties of the graph. Additional information, like node's attributes or features, can be integrated into the embedding process by simultaneously considering feature clustering, dimensionality reduction, and graph embedding into a global objective function that combines different goals. 

However, the dfRge approach is a technique that generates the pairwise vertex distance, and {\em per se} can only be compared with spectral embedding or techniques using vertices attributes. In addition, as explained before the DRC can combine both structural properties (neighborhood distance structure), and vertex attributes (edge weights) into a single framework with needing complex multi-objective optimization techniques.

This motivates why in the evaluation section we compare dfRge with spectral embedding alone, as almost all reported approaches in the literature take as input the pairwise vertex distance.

\section{Validation}\label{sec:validation}
We validate the convergence of the discrete Ricci flow to a constant discrete curvature and the concepts developed in the previous sections over a set of specific graphs and evaluate the resulting geometric properties. For each graph we have calculated the discrete Ricci Flow Embedding (dfRge), along with the spectral embedding proposed in \cite{ROBLESKELLY20071042}, that is based on spectral projection, {\em i.e.}, a Multi-Dimensional Scaling (MDS), of the Laplacian matrix of the graph. This last embedding is chosen as a reference, as it uses the relationship between the graph Laplacian and the Laplace–Beltrami operator over the manifold \cite{10.5555/2980539.2980616}. This relationship is also at the core of Graph Neural Network \cite{NEURIPS2021_0cbed40c}. The results of the comparison with spectral embeddings extend to a wide variety of embeddings based on the Laplace–Beltrami operator. As explained in the section \ref{sec:embedding}, other techniques used in the literature as isomap are not useful as a comparison as they begin with the distance matrix, which is the outcome of the dfRge.

As the edge distances are rescaled to an average of 1, its standard deviation evaluates the dispersion around this average. We also use the average ORC and its standard deviation as a  metric showing the convergence of the DRC. To check the robustness of the results, all experiments are executed 10 times with different random seeds. 

For validation purposes, we are using 5 different graphs described below. All these graphs have $N=10,000$ vertices.   

\textbf{Test graphs:} 
\begin{itemize}
    \item Random 3-regular graph: Graph in which each vertex is connected randomly to 3 other vertices.
    \item Erd{\"o}s-Reyni graph: Random graph in which two vertices are connected with probability $p$. We set $p$ a little above the connectivity threshold $p^*=\frac{log N}{N}$, {\em i.e.}, $p=9.4 \times 10^{-4}$.
    \item Plane random geometric graph: we generate randomly and uniformly a set of 10000 points, in the range $0\le x,y\le 10$ over a plane in the 3-dimensional space. We connect points that are at a distance at most equal $0.02$. This results in a connected graph.
    \item 5-cylinder random geometric graph: We generate randomly and uniformly a set of 10000 points over a cylindrical surface, then we connect each point to its $k=5$ closest neighbors.
    \item Stochastic Block Model graph (SBM): This type of graph is closely related to Erd{\"o}s-Reyni graphs. It combines several Erd{\"o}s-Reyni graphs with different connectivity parameters \cite{HOLLAND1983109}. These graphs are used as benchmarks for community detection algorithms. It consists of a partition of vertices into $k$ sub-sets called communities. Vertices in a community $i$ are connected together with probability $p_{ii}$ and to vertices in other communities with probability $p_{ij}<p_{ii}$, {\em i.e.}, vertices in one community are more connected together than into other communities. There is a rich literature on theoretical limits on $p_{ii}$ and $p_{ij}$ enabling easy retrieval of vertices in each community. In particular, the Kesten-Stigum (KS) threshold delimits the setting where the retrieval of blocks is easy \cite{Abbe2015}. To evaluate the discrete Ricci flow embedding, we generated an SBM with three equal-sized communities, with intra-community probability  $p_{ii}=0.01$ and inter-community probability $p_{i,j}=0.003$.
\end{itemize}
For all these five types of graphs, we calculated the dfRge after 20 iterations. The table \ref{tab:tabvalid} shows the results of the convergence of the DRC  in terms of average curvature, standard deviation of average curvature, average distance, and standard deviation of average distance. We can observe that the dfRge is keeping its promise in projecting the graph in a constant curvature, with very small $\sigma$ values, space where geometric analysis is doable. Nonetheless, spectral embedding fails to project these graphs into a fixed curvature space,  and therefore inferences made over geometric arguments with these embeddings' should be taken with caution. Interestingly, $k$-regular and Erd{\"o}s-Reyni graphs, where all vertices are statistically equivalent, generate dRfge that have uniform distance between neighbors, as can be expected from theory. However, Spectral embedding fails to detect this uniformity, {\em i.e.}, a clustering method applied to the spectral embedding distances will generate several well-separated clusters, while an Erd\"{o}s-Renyi should only generate a single cluster \cite{5494932}. The spectral embedding generates geometric artifacts that were not present in the initial data. A more detailed analysis shows that the spectral embedding distances are strongly correlated with the degree of the link extremities, {\em i.e}, the larger the degree the more distant the nodes, and clusters will be built around nodes with higher degrees. These values are very close to the KS threshold, and it is well-known for such values, spectral techniques cannot separate well the communities. 

However, for cylindrical random geometric graphs, vertices are not statistically equivalent, even if they are all the same degree. The dRfge acknowledges that with creating a diverse distance distribution with $\sigma=1.53$. Moreover, the average curvature attained is very small, leading to the intuition that the dRfge is detecting the flatness of the original cylinder, which is very clear in the case of the plane Random graph, the Ricci flow leads to flat embedding with curvature of -0.08(0.007).

For the SBM graph we see that DRC is converging, not as well as for the Erdos-Reyni. The distance exhibits large variations. A more detailed analysis for the results of the SBM graph shows that vertices in block 1 have an average distance of 1.05, in block 2 of 1.085, and in block 3 of 0.97. Moreover, the vertices in blocks 1 and 2 are distant on average by a distance of 1.6, block 1 and 3 vertices by an average distance of 2.3, and block 2 and 3 vertices by an average distance of 1.7. This means that the Ricci flow embedding retrieves the particular geometry of the 3-block SBM at values close to the KS threshold.  This is the principal motivation behind the proposition to use Discrete Ricci Flow for graph clustering \cite{osti_10153644,10.1093/imaiai/iaaa040}. 
However, differently from \cite{osti_10153644,10.1093/imaiai/iaaa040} here, we do not apply surgery (as it is not needed for discrete Ricci Flow). 

\begin{table}[h]
  \centering
  \small
  \begin{tabular}{|l|l|r|r| }
   \hline
   \multicolumn{1}{|l|}{\textbf{Graph type}}
   & \multicolumn{1}{|l|}{\textbf{Embedding}} & \multicolumn{1}{|l|}{\textbf{Average Curvature~(std)}}  & \multicolumn{1}{|l|}{\textbf{Average Distance~(std)}}\\
    \hline
    \multirow{2}{10em}{Random $3$-regular graph}& 
    Spectral & -0.698~(0.255) & 1.0~(0.170)\\ 
    \cline{2-4}
    &Ricci & -0.682~(0.007) & 1.0~(0.0075)\\ 
    \hline
    \multirow{2}{10em}{Erd\"{o}s-Renyi graph $\bar{k}=9.39$}& 
   Spectral & -1.697~(1.000) & 1.0~(6.116)  \\  \cline{2-4}
   &Ricci & -1.38~(0.017) & 1.0~(0.082)\\ 
    \hline
    \multirow{2}{12em}{plane Random Geometric Graph with $\epsilon=0.02$-balls}& 
    Spectral & -0.36~(0.875) & 1.0~(1.2) \\ 
    \cline{2-4}
    &Ricci & -0.08~(0.007) & 1.0~(1.28)\\ 
    \hline
    \multirow{2}{12em}{Cylinder Random Geometric Graph $k=5$}& 
    Spectral & -0.241~(0.918) & 1.0~(0.830) \\ 
    \cline{2-4}
    &Ricci & -0.172~(0.018) & 1.0~(1.530)\\ 
    \hline
    \multirow{2}{12em}{Stochastic Block Model
Graph}& 
    Spectral & -1.84(1.23) & 1.0(0.87)  \\ 
    \cline{2-4}
    &Ricci & -1.43 (0.042)  & 1.0 (0.43)\\ 
    \hline
  \end{tabular}
  \caption{Embedding results for different types of graphs all with $N=$10,000 nodes. The values shown are average (standard deviation)}
  \label{tab:tabvalid}
\end{table}
Overall the validation shows that the DRF provides more consistent distance embeddings than the spectral method. This is not unexpected, taking into account the theoretical arguments given in the paper. On another hand the validation shows that spectral embedding method that is at the core of almost all graph embedding in literature project the graph into a space  where the curvature are not constant, resulting into biases when geometric methods are used.

\section{Accelerating the Calculation of the Ollivier-Ricci Curvature}\label{sec:algorithms}
Section~\ref{sec:theory} discussed how deriving the DRC for an edge involves solving a linear optimization problem of optimal transport, that takes as input a distance matrix of shortest paths between neighbors of the edge extremities and a source and destination distributions over these neighbors. While solving this linear program is relatively of low complexity and can be derived in practice in milliseconds, the principal complexity challenge is the derivation of the distance matrix that can involve, for a single iteration of the discrete Ricci flow,  hundredth of millions of shortest path calculation for practical graphs. A naive implementation could require tens of hours of calculation to derive the DRC for all edges of a  graph of reasonable size. 

This large complexity explains why previous research works using Ricci curvature have evaluated it on relatively small graphs (in order to hundreds of nodes), and why this approach has not gained thrust in the machine learning community. Using the discrete Ricci flow embedding in practice entails finding ways to reduce this complexity. A contribution of this paper is to reduce the height of this tractability wall through some algorithmic shortcuts.

\subsection{Single Source-Multiple Destination Dijkstra Algorithm}
In regard to its importance, a relatively large literature has been devoted to optimizing the calculation of shortest paths. The Single Source to Single Destination (SSSD) problem with positive edge lengths can be solved using the famous Dijkstra algorithm with a linear complexity $\Omega(N)$ \cite{EWD:EWD316}. It consists of maintaining an ordered heap containing neighbors of nodes visited through a Breadth First Search (BFS) beginning at the source and removing from the heap, at each iteration, settling the shortest path distance from the source of the vertex at the top of the heap and removing it. The process continues till the destination appears at the top of the heap. While the shortest path distance complexity is $\Omega(N)$, in practice, the real complexity might be lower, depending on the density of the graph and its node degree distribution. Some variants of the Dijkstra algorithms, like  $A^*$ algorithm \cite{Botea2004NearOH}, can be used when additional information is available. A closely related problem is the All-Source to All-Destination (ASAD) shortest path derivation that can be solved by the Floyd-Warshall algorithm. This algorithm finds, with a complexity $\mathcal{O}(N^3)$, all distances in a graph by incrementally improving an estimate on the shortest path between each two vertices until the estimates are optimal. The Floyd-Marshall algorithm needs to maintain a $N\times N$ dense distance matrix with a large memory footprint for large graphs. 

Deriving the distance matrix needed for calculating the DRC entails solving a particular instance of the Multi-Source, Multi-Destination (MSMD) shortest paths where all source-destination pairs share a part of their neighborhood. We can leverage this property to break the complexity. Let's look at a simple case with a single source and multiple destinations. One can run multiple instances of SSSD all with the same source but different destinations. But this will go, in each run, through the same vertices and waste time. Another approach, called Single Source-Multiple Destination (SSMD) shortest path, consists in continuing to process the ordered heap till all destinations in a set are found. This approach is more efficient, as it integrates all pairs that are likely to share the same segments into a single run of the shortest path algorithm, avoiding processing the same vertices. 
\begin{figure}[ht]
  \centering
        \includegraphics[width=0.6\textwidth]{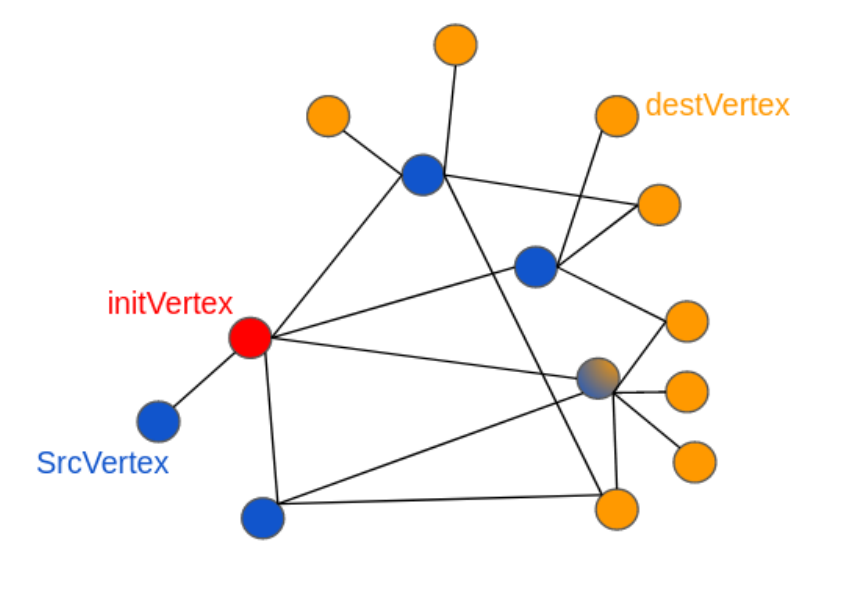}
        \caption{Arrangement of vertices to derive efficiently the distance matrix in curvature calculation} \label{fig:arrangement}
\end{figure}
\subsection{Vertex arrangement in tasks}\label{sec:arrange}
We can further benefit from the specific structure of the distance matrices needed to calculate the DRC, by arranging the source and destination sets in a particular way. Let's consider a vertex {\tt initVertex} in the graph, and let's consider the set of neighbors of {\tt initVertex} as the sources set $\mathcal{S}=\{s_1,s_2,\ldots,s_k\}$. The neighbors of the vertices in the sources set are considered as the destination set $\mathcal{D}=\{d_1,d_2,\ldots,d_l\}$. All distance matrices that are needed to derive the optimal transport over all edges adjacent to {\tt initVertex} can be retrieved as a sub-matrix of the distance matrix $\mathcal{S} \times \mathcal{D}$. We define the calculation task $\mathcal{T}$({\tt initVertex}) relative to vertex {\tt initVertex} as, calculating the distance matrix $\mathcal{S}\times \mathcal{D}$ defined above. At the end of the task $\mathcal{T}$({\tt initVertex}) the optimal transport optimization problem can be solved for each edge adjacent to {\tt initVertex} and its DRC can be obtained. We show the arrangement in Fig. \ref{fig:arrangement}.

Therefore, to derive the DRC for all edges in a graph, we can go through the list of all vertices $v_i$ and calculate the task $\mathcal{T}(v_i)$. Whenever an edge adjacent to $v_i$ has already been processed, we are not adding it to the task $\mathcal{T}(v_i)$  and remove the related destinations from the set $\mathcal{D}$. Tasks are created until all edges in the graph have been processed. The size of each task to run, {\em i.e.}, the size of the distance matrix, $|\mathcal{S}\times |\mathcal{D}|$, depends directly on the degree distribution of the graph, while the number of tasks depends on the density of the graph.  Moreover, the distance matrices involve vertices that are at most 3 hops apart. Therefore, the performance depends on the size of the 3-hops neighborhood through controlling the number of elements inserted into the ordered heap which is at the core of the Dijkstra algorithm. Thus, the properties of the graph impact strongly the execution time. It is noteworthy that each one of the above-defined tasks is independent and can be executed in parallel. This is precisely what we do to speed up the process. Each task is entered into a job queue and processed whenever a thread from a thread pool is available.  This relatively simple implementation is very efficient. We implemented a C++ version of the above-described algorithms.

\subsection{Performance Evaluation}
In order to evaluate the performance of the presented algorithms, we use two synthetic graphs, with similar number of vertices: an Erd{\"o}s-Renyi (ER) Graph with 10,000 vertices and 93,919 edges, and a connected Watts-Strogatz (WS) small world graph with 10,000 vertices and 60,000 edges. We are also using a real AS-level BGP graph representing the whole Internet with 77,804 vertices and 280,231 edges. These three graphs are used as they have different node degree distributions and densities. 

We implemented two parallel methods for deriving the needed distance matrices: SSSD, and SSMD. In addition, we also used two approaches for defining a calculation task. In the first approach, the arrangement described in Sec. \ref{sec:arrange} is not used, and each task is calculating the curvature of a single edge by first deriving the relevant neighbors' distance matrix and by using it to derive the DRC of the edge. The second approach uses the arrangement proposed in Sec. \ref{sec:arrange}. We show in table \ref{tab:MSMD} the overall running time in seconds for one iteration of calculating the DRC for all edges of the graph. We show the results for both SSSD and SSMD for the second task definition approach and the results for SSMD alone for the first approach. We are not showing the SSSD results with the first approach, as they are clearly non-competitive,  {\em e.g.}, the SSSD with the first approach over the BGP needs several days of calculation. 

All experiments are made using a computer equipped with an AMD EPYC 7401P 24-core Processor with 48 threads, 64 MB of L3 cache, and running at a 3GHz clock rate. The installed RAM is 256GB. The in-use memory footprint never exceeded 200 MB in our experiments, thus we do not discuss it further.
\begin{table}[ht]
\small
\centering
\begin{tabular}{|l|r|r|r|r|}
\hline
\multicolumn{1}{|p{1.5cm}|}{\textbf{Graph type}} & \multicolumn{1}{p{2.5cm}|}{\textbf{Shortest paths~(millions)}} & \multicolumn{1}{p{2.9cm}|}{\textbf{SSMD without vertex arrangement~(s)}} & \multicolumn{1}{p{2.6cm}|}{\textbf{SSSD with vertex arrangement~(s)}} & \multicolumn{1}{p{2.6cm}|}{\textbf{SSMD with vertex arrangement~(s)}} \\
\hline
ER graph& 39.4& 811.43& 3404.61 & 41.69  
\\ \hline
WS graph& 9.2& 63.79&95.45 & 5.39
\\ \hline
BGP graph & 854&5653.25&456808.12& 200.47
\\ \hline
\end{tabular}
\caption{Calculation time for distance matrices needed to derive the curvature of all edges of 3 types of graph using three alternative methods.}
\label{tab:MSMD}
 \vspace{-8mm}
\end{table}  
The results are summarized in Table~\ref{tab:MSMD}. The first conclusion to be drawn from the table is that the SSMD method is much more efficient than the SSSD method, with speedup gains ranging from 19 (for the WS graph) to 2284 (for the BGP graph). Moreover, the arrangement scheme described in Sec. \ref{sec:arrange} brings by itself a gain ranging from 12.6 (for the WS graph) to 28.16 (for the BGP graph). This is the combination of these two algorithmic tweaks that make the dRfge tractable even for large graphs as the BGP one.

It is noteworthy that the Ricci flow needs to execute several iterations of the DRC calculation over the graph. In the experiment done in this paper, we needed less than 10 iterations to get to a standard deviation below 0.02. This means that the expected running time for dRfge is in the order of 10 times the values reported in Table~\ref{tab:MSMD}. Last but not least, we evaluated the impact of the complexity of solving the linear program of optimal transport, in all the experiments that we did, the time spent on solving this program never exceeded 0.01\% of the total derivation time. So we did not bother to use an approximation to the optimal transport like in \cite{Cuturi} to speed up the calculations.

\section{Use Case: BGP Graph Analysis}\label{sec:applications}
Previous sections introduced the discrete Ricci embedding for graphs and motivated the use of this embedding using arguments relative to the correctness and possibility of the geometric interpretation. We also validated the approach on synthetic graphs. In this section, we present an application of our method on an AS-level BGP graph representing the whole Internet \cite{10.5555/839292.842983} with 77,804 vertices and 280,231 edges captured in July 2022. In \cite{7218668}, the authors show that internet topology through the application of $0$-curvature to an simplified Internet graph. Here we look at a global large scale Internet graph after DRC convergence. The discrete Ricci flow over this large-scale graph converges after 17 iterations to a constant curvature space with curvature equal to -0.32(0.01) and edge distances in the range $[0.0007,24.3]$.
\begin{figure}[ht]
  \centering
  \includegraphics[scale=0.45]{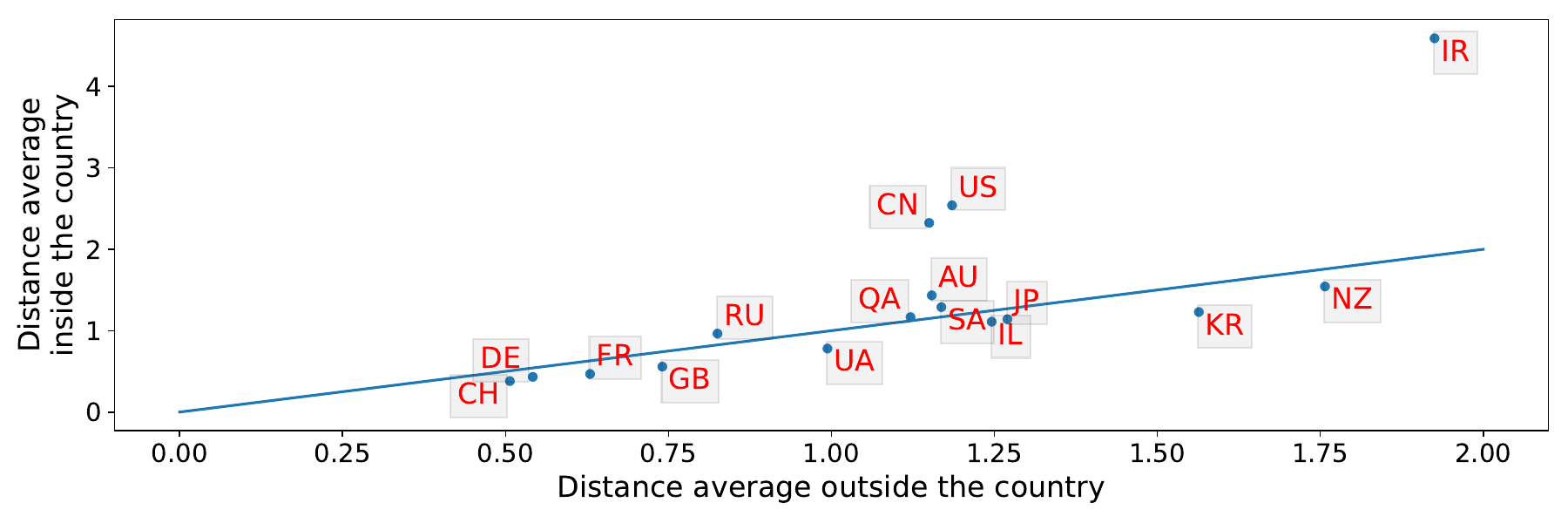}
  \caption{Ratio of the average edge distance inside a country over edge distance to neighboring countries} \label{fig:ratio}
\end{figure}
Over this graph, we aim to analyze the global structure of connectivity between countries. In particular, we evaluate how the external Internet connectivity of a country compares with its internal connectivity, comparing the average distance in a country to the average distance to neighboring countries provides insights into bottlenecks residing within or outside a given country. For this purpose, we need a geometric representation of the connectivity graph that evaluates the relative strength of the bottleneck created by each link. With such a metric, it becomes possible to compare intra-country bottleneck strength vs. inter-countries and gather some insights on the strategy of different countries. The distance metrics obtained from the DRF discrete Ricci flow are a good candidate metric for this objective. These metrics not only represent the local bottlenecks, but they evaluate the impact of downstream an upstream bottlenecks thank to the the non-linear regularization done by the Ricci flow.

It is noteworthy that because we have embedded the BGP graph into a constant curvature  space, it becomes meaningful and correct to compare distances between edges.

We show in Fig. \ref{fig:ratio} the values of the average edge distance inside a country vs. the average edge distance to the neighboring countries along with the equality line for a group of representative countries. A value below the equality line  means that the bottlenecks outside the country are stronger on average than the ones in the country. The opposite applies to values above the line. We observe that points relative to European countries, along with Japan and Korea, are below the line, while countries implementing a censorship control over their outside traffic, like Iran, Russia, Saudi Arabia, Qatar, or China are above the line. This can be explained by the fact that these countries create artificial connectivity bottlenecks inside their country to redirect traffic to control points. The large value in US can be explained by the fact that multinational US based corporate with large-scale international internal networks like Google, Microsoft, Amazon, or Akamai are considered to be inside the US. The case of Australia is intriguing as it cannot be explained by the island's nature (as Japan, Great Britain, and New Zealand have values below one).


\section{Conclusion}
In this paper, we argue that to get a correct geometrical interpretation of graphs, we need to adopt embeddings that result into regular constant curvature space. We present the discrete Ricci flow Embedding as a solution that guarantees this property. We prove that the discrete Ricci flow converges to a constant curvature for all graph edges and a stable distance metrics, providing an equivalent of the Perelman-Poincare theorem for graphs. We show that differently from 3-manifold the convergence will be to regular constant curvature space and that surgeries are not needed.

Further, we present a new set of algorithmic shortcuts to calculate th e DRC over large-scale graphs and show that they can reduce computing time by two to three orders of magnitude over conventional algorithms. Finally, we illustrate the use of the approach on synthetic graphs, as well as validate it on a realistic case study. 

This paper opens major perspectives in several directions: Machine Learning indeed, resilience of graphs (as illustrated in the use case), and more generally analysis of graphs. The next steps will be to apply the embedding in clustering, graph inference and topological analysis of graphs. There is already a rich literature on the application of Ricci flow to these problems. Notably, this paper provide strong theoretical basis for these approaches and correct the curvature formulation they use.


\bibliographystyle{plain}
\bibliography{sample-base}
\newpage

\appendix







\section{Theoretical considerations}
We provide here more details on the theoretical roots of the method proposed in the paper. It is noteworthy to emphasize that our aim here is to give to the readership the intuition of the mathematical concepts that are behind the practical tools we will introduce in the paper. We have therefore balanced the mathematical precision with the ease of reading. A complete and precise presentation of differential and Riemannian geometry materials can be found in \cite{lafontaine} and in the given references there.
\subsection{Curvature and Ricci flow}
\label{sec:curva}
Different notions of curvature have been defined in differential geometry,  that all represent and quantify the deviation of the space from flatness. A surface might be approximated by a tangent hypersphere above the surface resulting into a positive curvature;  a tangent hypersphere below the surface resulting into negative curvature; a tangent plane (a hypersphere with radius $\infty$) resulting into a null curvature. Curvature has a central position in the theory of manifold, as it relates the global structure of the space and its topology to local and microscopic properties of the manifold. The famous Gauss-Bonnet theorem \cite{lafontaine}:
\begin{equation}
    \int_{M}\Gamma dA+\int_{\partial M} \Gamma_g ds=2\pi\chi(M)
\end{equation}
relates the integral of the Gaussian curvature $\Gamma$, a local and microscopic property, over a compact 2-dimensional Riemannian manifold $M$ and its border $\partial M$,  to its Euler characteristic or genus (the number of holes in the surface), $\chi(M)$. This theorem provides an illustration of the relationship between microscopic properties of the manifold, the curvature, and global and fundamental topological property, the genus.

The Ricci curvature, $\kappa$, extend the concept of curvature by looking at the behavior of geodesics, curves with the shortest metric distance over the manifold, between two points. More precisely, let $x$, and $y$ be two  close points over a manifold, and let's look at geodesics coming out of these two points. When curvature is positive, resp. negative,  geodesic get farther apart, resp. closer. Ricci's curvature assess the speed of divergence or convergence of the geodesics coming out of a point. In more precise terms, let's consider a ball of radius $\epsilon$ around $x$, and transport all points of this ball along the relevant geodesics to a ball around a close-by point $y$. Then, on average, points in this ball will travel less than the distance between $x$ and $y$ when the curvature is positive, and travel more, when the curvature is negative. Ricci and Gaussian curvature are closely related, {\em e.g.}, for 2D spaces, $\Gamma=2\kappa$.

The Ricci flow over a Riemannian manifold $M$, introduced in \cite{Hamilton88}, is a non-linear Partial Differential Equation (PDE) that acts on a metric $g(t)$ (that can be any function with metric properties) defined over the manifold, proportionally to the Ricci curvature $\kappa^c(g(t))$ that depends itself on the metric $g(t)$:
\begin{equation}
\label{eq:ricciflow}
\frac{\partial g(t)}{\partial t}=-2\kappa^c(g(t))   
\end{equation}
The Ricci flow reduces the distance metric $g(t)$ in the regions with positive curvature, and increases it in the regions with negative curvature. The changes in the distance metric results in flattening the manifold by changing the curvature $\kappa^c(g(t))$. This means that the Ricci flow acts similarly to a steam iron that is pressing over the complex manifold wrinkles and transforming it into a flat sheet. The Ricci flow is a major tool in proving the Poincar{\'e}-Perelman theorem\cite{Perelman2003}. We refer the reader for a more complete explanation of the mathematical details to \cite{lafontaine,Hamilton88}.

\subsection{Geometric analysis}
Geometry has been invented 5000 years ago for solving concrete spatial problems, like, calculating unknown distances, angles, areas, or volumes from known values. For example, calculating the third side of a triangle using two sides has been the main motivations behind famous geometry theorems like the Thales and Pythagorean ones. The most general theorem on this topic is the Al-Kashi-Toosi theorem \cite{1891traite}, also known as the cosine theorem,  that relates the length of a side of a triangle to the two other lengths and the opposite angle 
\begin{equation}
    c^2=a^2+b^2-2ab\cos \gamma
\end{equation}
This theorem is the main element behind the definition of Euclidean distance and the whole "geometric intuition" that is central to most machine learning algorithms. For example, $k$-means algorithms consists of calculating the distance of points to several centro{\"i}ds representing clusters and choosing the closest one. The centro{\"i}d or center of gravity of $k$ points is the point that minimizes the sum of squared distances between itself and each point in the set. These distances are evaluated using a specific version of cosines theorem: the famous Pythagorean theorem. 

The cosine theorem can be extended to constant curvature and regular manifolds. For a manifold with constant negative Gaussian curvature, $\Gamma=-\frac{1}{k^2}$, we have : 
\begin{equation}
    \cosh{\frac{c}{k}}=\cosh{\frac{a}{k}}\cosh{\frac{b}{k}}-\sinh{\frac{a}{k}}\sinh{\frac{b}{k}}\cos\gamma,
\end{equation}
that reduces to classical cosines law in Euclidean space when $\Gamma \rightarrow 0$
($k\rightarrow \infty$). The equivalent of the Pythagorean theorem in a hyperbolic space becomes $\cosh{\frac{c}{k}}=\cosh{\frac{a}{k}}\cosh{\frac{b}{k}}$.

\begin{figure}[ht]
  \centering
  \includegraphics[scale=0.2]{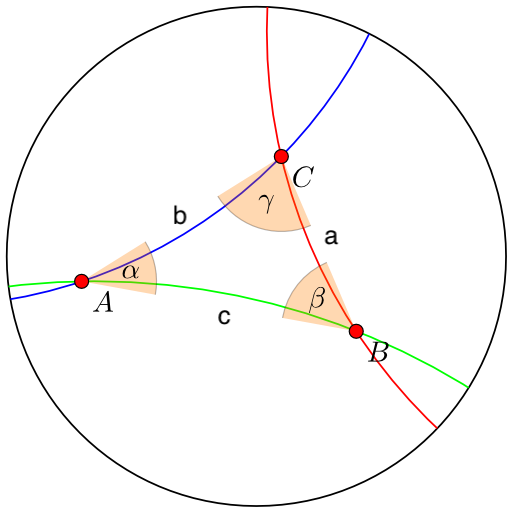}
  \caption{Cosine law for constant curvature hyperbolic space} \label{fig:ratio}
\end{figure}
Through these extensions, geometric analysis is possible over constant curvature hyperbolic space, and it is possible to evaluate unknown distance through triangulation. In particular, centro{\"i}d calculation can be done using the method described in \cite{2021}. Similar cosine law exist for spherical constant curvature space with positive curvature. 

However, for spaces with non-constant curvature (general manifold), there is no equivalent of law of cosines, {\em i.e.}, for a triangle laying over a general manifold, knowing two sides and the angle between them does not make possible to evaluate the length of the third side. The best that can be achieved is given through the Toponogov's theorem. This theorem states that for a Riemannian manifold with Gaussian curvature $K$ satisfying $K\geq \delta$,  and for a triangle ABC with known side $a$, $b$ and angle $\gamma$,  the length of the third side $c$, will be less than the length $c'$ of a triangle with two sides length equal to  $a$, $b$ and angle $\gamma$ in a constant curvature space with curvature $\delta$. Because of the exponential growth of the $cosh()$ function, this bound is not tight and results in serious overestimation of distances. In  particular, assuming Euclidean distance for points that are projected over a negative curvature manifold will be flawed. 

Solving the above issue over a non-constant curvature manifold need to have the detailed microscopic characteristics of the manifold to calculate precisely the unknown distances through differential geometry formulas. This is frequently impossible to be done as these precise microscopic information are not available. Only over constant curvature space, we can related distance and angles to each other. The Ricci Flow is precisely mapping a non-constant curvature manifold where geometric intuitions are not valid to an constant curvature one, and it is therefore as the core of the approach developed in the paper.

The previous observation about overestimation of distance over general manifold, has major implications in machine learning as a wide scope of techniques on graphs or on point data, like spectral embeddings, $t$-SNE, or Graph Neural Networks, are using embeddings that project the initial data into general manifolds where as said above geometry is not valid.
\section{Proofs}

\subsection{Proof of Lemme \ref{lemma:2}}
\label{appendix:proof}

\begin{lem}
Let $G=(V, E,w)$ be a weighted graph, and $e=(x,y)\in E$ an edge in the graph. For any two distance vectors ${\bf d}$ and ${\bf d}'$, such that ${\bf d}$ and ${\bf d}'$ differ only on edge $(x,y)$ and $d(x,y) \leq d'(x,y)$, and for any function $g\in Lip(1)$, such that $g(y)-g(x) = d'(x,y)$, there exists a function $g' \in Lip(1)$, such that $g'(y)-g'(x) = d(x,y)$, and $d(x,y) - d'(x,y)\leq g'(z)-g(z) \leq 0$ for all $z$ in $V$ .
\end{lem}

\begin{proof}
    The lemma states that for a function $g \in Lip(1)$, that validates the distance constraint on edge $(x,y)$, $g(y)-g(x)=d'(x,y)$, there exists an approximation $g'$ in $Lip(1)$, that validate the distance constraint, $g'(y)-g'(x)=d(x,y)$.
    
    We know that $g$ is in $Lip(1)$, so $|g(s) -g(t)| \leq d'(s,t)$ for all $(s,t)$ in $E$. 
    Now let's define the function $\mathcal{G}$ as:
    $$
        \mathcal{G}(z) = \left\{
        \begin{array}{ll}
            g(x) + d(x,y) & \mbox{if}\; z = y \\
            g(z) & \mbox{otherwise}
        \end{array}
    \right.
    $$
    We have $\mathcal{G}(y) - \mathcal{G}(x)=d(x,y) $, so $\mathcal{G}$ validate the distance constraint on edge $(x,y)$. Moreover, $g(y)-\mathcal{G}(y)  =g(y)- g(x)+d(x,y) = d(x,y)-d'(x,y)$. And, $g(z)-\mathcal{G}(z)=0$  for $z\neq y$. 
    
    The edges $(z,y)$ are the only one to get directly impacted by the change of value of $\mathcal{G}$ compared to $g$.  So  let's consider first $\mathcal{G}(z)-\mathcal{G}(y)$ for $z \in N(y)$. If the function $\mathcal{G}$ is not in $Lip(1)$ there exists a set of $n$ neighbors of vertex $y$  $\{z_1, \ldots, z_n\}$, indexed such that $\mathcal{G}(z_1)\le \mathcal{G}(z_2)\le \ldots\le \mathcal{G}(z_n)<\mathcal{G}(y)$, and  $|\mathcal{G}(z_i)-\mathcal{G}(y)|>d'(z_i,y)$. Let $Z=\{z_1,\ldots,z_n,y\}$. It is notable that function $\mathcal{G}$ is only different from $g$ on vertex $y$, therefore it inherits from $g$ the $Lip(1)$ property for all other edges beyond the ones with extremities in $Z$. 
    
    Moreover, we have $|\mathcal{G}(z_i)- \mathcal{G}(x)|$ not equal to $d(x,z_i)$, this results from $\mathcal{G}(z_i)-\mathcal{G}(x) > d(y,z_i)+ +\mathcal{G}(y)-\mathcal{G}(x) = d(y,z_i)+d(x,y) >-d(x,z_i)$ and $\mathcal{G}(z_i)-\mathcal{G}(x) <-d(y,z_i) +d(x,y) <-d(y,z_i)+d(x,y) < d(x,z_i)$. Thus subtracting an appropriate $\delta_i>0$ from  $\mathcal{G}(z_i)$ will not affect $(x,z_i)$.
    
    Now we create function $\mathcal{G}'(z_i)=\mathcal{G}(z_i)-\delta_i$ so that $\mathcal{G}(z_j)-\delta_j \le \mathcal{G}(y) + d(y,z_i)$. $\delta_i$ is lower bounded by $\mathcal{G}(z_j)- g(y)- d(y,z_i) > 0$, and upper bounded by $d'(x,y)-d(x,y)$.
    
    To ensure Lip(1) property among edges $(z_i,z_j)$, we let $\forall i,j,\; 1\le j\le i\le  n-1,\; < \delta_j-\delta_i < \mathcal{G}(z_j)-\mathcal{G}(z_i)+d(z_i,z_j)$. We can ensure that by decreasing the values of $\mathcal{G}(z_i)$ as much as possible in the reverse order, from $z_n$ to $z_1$. Now all pairs among $Z$ are $Lip(1)$ under $\mathcal{G}'$. Moreover with, $\delta_i \ge \mathcal{G}(z_i)-\mathcal{G}(x)- d(x,z_i)$ for each $i$, the new function $\mathcal{G}'$ is still  in $Lip(1)$ for pair $x$ and $z_i$.

For any other vertices $v$ out of set $Z \cup \{x\}$, if $\mathcal{G}(z_i)- \mathcal{G}(v) =-d(v,z_i)$ , we need to subtract $\delta_i$ from $g(v)$ so that $\mathcal{G}'$ is $Lip(1)$ for edge $(v,z_i)$. Nonetheless, such $v$ does not exist, as it would lead to $\mathcal{G}(v)- \mathcal{G}(y) = \mathcal{G}(z_i) + d(v,z_i) > d(y,z_i) + d(v,z_i) \le d(y,z)$ that is a contradiction. If $\mathcal{G}(v)- \mathcal{G}(z_i) = d(v,z_i) , \mathcal{G}(v)- \mathcal{G}'(z_i) = \mathcal{G}(v)- \mathcal{G}(z_i)- \delta_i < d(v,z_i)$ thus there is no need to reduce value from $\mathcal{G}(v)$. We have similar conclusion for cases where $|\mathcal{G}(v)-\mathcal{F}(z_i)| <d(v,z_i)$ . Finally, there exist $\delta_i>0$ so that $\mathcal{G}'(z_i)=\mathcal{G}(z_i)-\delta_i$ is $Lip(1)$ over all edges of the graph $G$, and $\mathcal{G}'(y)- \mathcal{G}'(x) = d(x,y)$. Moreover, $\mathcal{G}'(z)-g(z)=0$ for vertices $z \notin Z$ and $0\geq \mathcal{G}'(z_i)-g(z_i)=-\delta_i \geq d(x,y)-d'x(y)$ for all $z_i \in Z$. The proof is complete.
\end{proof}

\section{Additional Use Case: E-roads map analysis}
\label{Appendix:useCase2}
An additional use case we considered for illustrating the interest of the discrete Ricci embedding is the analysis of the network of all E-roads, major European roads numbered from E1 up and crossing national borders up to Central Asian countries like Kyrgyzstan in mainland European. The graph consists of $1,041$ vertices representing the main landmarks of the E-roads and $1,307$ edges. After 180 iterations, the discrete Ricci flow converges to a constant curvature space with curvature -0.501 (0.001) and edge distances in [0.0001, 7.89]. 


\begin{figure}
	\centering
	\begin{subfigure}[t]{0.95\linewidth}
		\includegraphics[width=\linewidth]{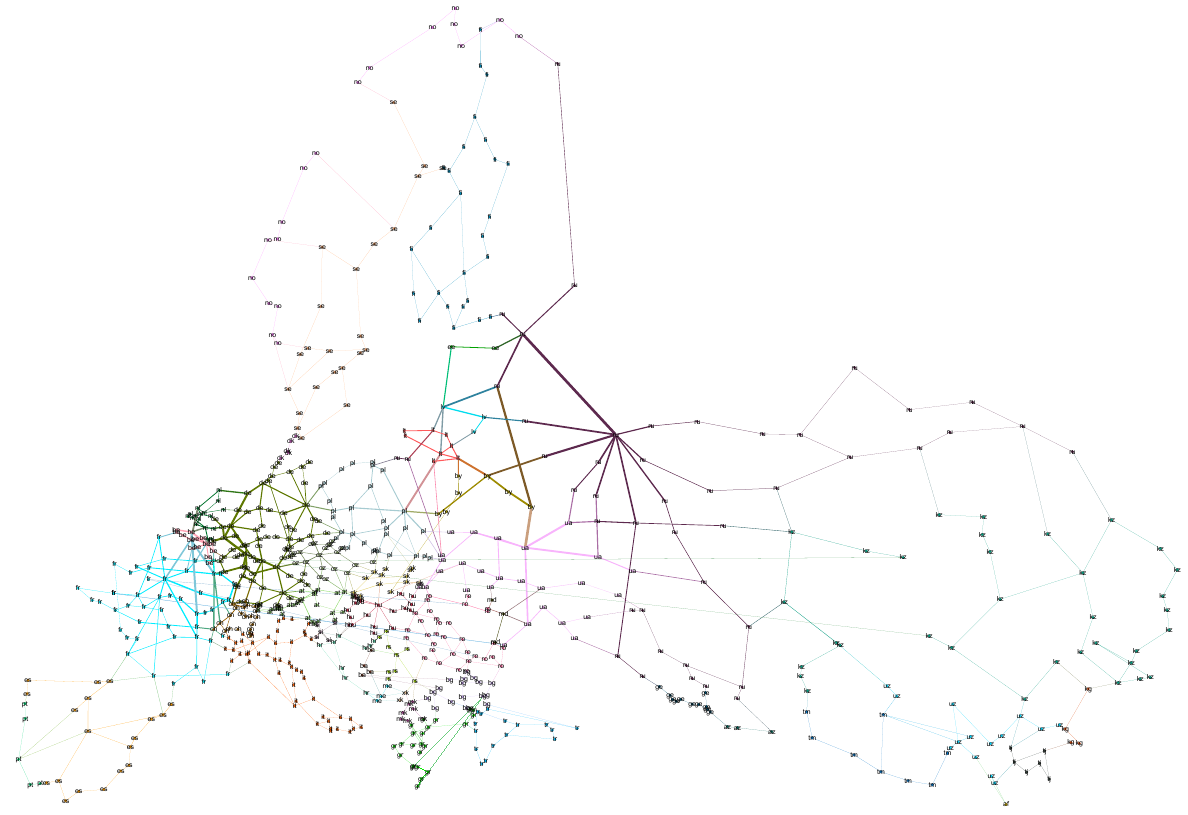}
		\caption{Complete European map}
		\label{fig:map}
	\end{subfigure}
	\hfill
	\begin{subfigure}[t]{0.95\linewidth}
		\includegraphics[width=\linewidth]{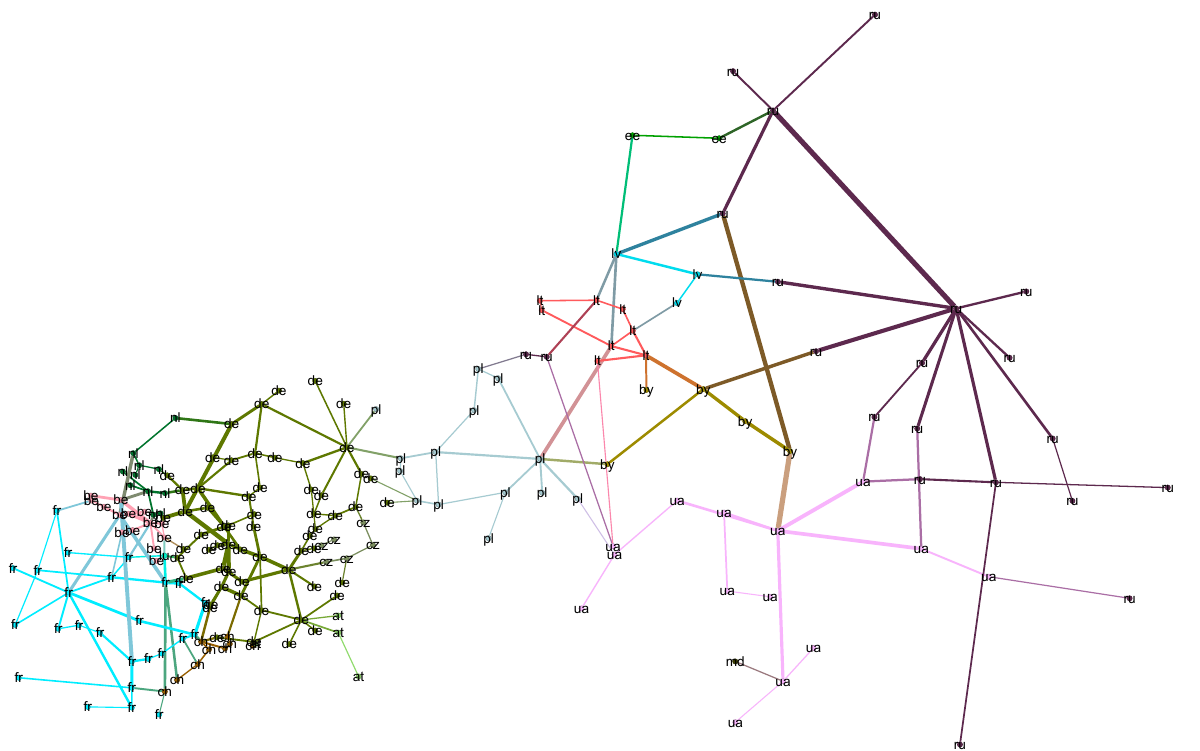}
		\caption{Map of European road filtered by edge distance larger than 1 after the Ricci flow}
		\label{fig:zoom2}
	\end{subfigure}
	\caption{The E-roads map after Ricci flow application (the intensity of the edge color is proportional to the Ricci flow distance equilibrium distance.}
	\label{fig:evaluation}
\end{figure}
We show in Figure \ref{fig:map} a map of the E-roads colored using distances obtained after the application of the Ricci Flow. The longest distance edges are shown as thicker  lines, while the shortest are shown as thin lines (the distances between points on the map are real physical distances). 

A classical problem in complex network analysis is to extract the network's backbone, {\em i.e.}, the sub-network that captures structurally important vertices and links. The backbone defines the mesoscale structure of the graph, by emphasizing long-range connections between communities.  In \cite{10.1093/comnet/cnac014}, it is suggested that the Ricci flows can be used to extract networks’ backbone. The map illustrates very well this ability of Ricci flows. We show in Fig. \ref{fig:zoom2}, the filtered graph where the 199 edges out of 1041 that have a distance larger than 1.2 are represented, the average length being 1.

This is noteworthy that as we are in a constant curvature space, the shortest path distances are geometrically meaningful and can be compared. The stretching and squeezing effect of the Ricci flow is clearly visible. Regions of the network that are separated from the rest of the network by geographical bottlenecks become squeezed (the edges width are thin), {\em e.g.},  the average edge distances inside the Scandinavian peninsula, (Sweden, Norway, and Finland vertices) is 0.01. 

Nonetheless, the large-distance edges are of utmost importance for the geometry of the graph, as they provide minimal distances for the regions of the graph they connect and belong to the graph backbone. Removing these edges will result in a different macroscopic structure for the graph, with distances that will change around these edges. 
The filtered graph in Fig. \ref{fig:zoom2} shows clearly that most of the long edges are in Central Europe  with the longest edge being the "Kiev-Gomel" edge with a length of 7.89. One can see two regions with high distances edges: one in in Russia and Ukraine and one that cover Germany (the average of distance in Fig. \ref{fig:zoom2} is 4.7). shown in Figure \ref{fig:zoom2}.  Indeed, most insights obtained by the analysis of the E-road network are well-known because we know the underlying geography of the map. However, it is remarkable that these insights can be retrieved using the graph alone and its discrete Ricci flow embedding.

\end{document}